\documentclass{amsart}
\usepackage[fontsize=11pt]{fontsize}
\usepackage{amsmath}
\usepackage{amsthm}
\usepackage{amssymb} 
\usepackage{hyperref}
\usepackage[capitalise]{cleveref}
\usepackage{mathrsfs}
\usepackage{dsfont}
\usepackage{color}
\usepackage{calrsfs}
\usepackage{epic}
\usepackage{graphicx}
\definecolor{red}{rgb}{.7,0,0}

\newcommand{\N}{\mathbb N}
\newcommand{\R}{\mathbb{R}}
\newcommand{\C}{\mathbb{C}}
\newcommand{\Z}{\mathbb{Z}}
\newcommand{\Q}{\mathbb{Q}}
\newcommand{\F}{\mathbb{F}}

\newcommand{\proj}{\mathbb{P}}

\renewcommand{\a}{\alpha}
\renewcommand{\b}{\beta}
\renewcommand{\c}{\gamma}

\newcommand{\e}{\varepsilon}

\newcommand{\cR}{\mathcal{R}} 
\newcommand{\cF}{\mathcal{F}} 
\newcommand{\la}{\lambda} 
\newcommand{\HN}{\mathrm{HN}}
\newcommand{\VP}{\mathsf{VP}}
\newcommand{\VNP}{\mathsf{VNP}}
\newcommand{\VPSPACE}{\mathsf{VPSPACE}}

\newcommand{\VCH}{\mathsf{VCH}}
\newcommand{\VPi}{\mathsf{V}\Pi\mathrm{P}}
\newcommand{\SP}{\#\mathsf{P}}
\newcommand{\PP}{\mathsf{PP}}
\newcommand{\Po}{\mathsf{P}}
\newcommand{\FPo}{\mathsf{FP}}

\newcommand{\NP}{\mathsf{NP}}
\newcommand{\AM}{\mathsf{AM}}
\newcommand{\PH}{\mathsf{PH}}
\newcommand{\CH}{\mathsf{CH}}
\newcommand{\GCC}{\mathsf{GCC}}
\newcommand{\VPnb}{\mathsf{VP}_{\mathrm{nb}}}
\newcommand{\VNPnb}{\mathsf{VNP}_{\mathrm{nb}}}

\newcommand{\PSPACE}{\mathsf{PSPACE}}
\newcommand{\BPP}{\mathsf{BPP}}

\newcommand{\poly}{\mathrm{poly}}
\newcommand{\per}{\mathrm{per}}

\newcommand{\sgn}{\mathrm{sgn}}

\newcommand{\cC}{\mathcal{C}}
\newcommand{\cD}{\mathcal{D}}

\newcommand{\chara}{\mathrm{char}\,}

\newcommand{\GL}{\mathrm{GL}}

\renewcommand{\det}{\mathrm{det}}

\newcommand{\Bit}{\mathrm{Bit}}

\newcommand{\Res}{\mathrm{Res}}

\newcommand{\NU}{\mathsf{nu}}
\newcommand{\U}{\mathrm{u}}
\newcommand{\adj}{\mathrm{adj}}
\newcommand{\tc}{\mathrm{tc}}
\newcommand{\hG}{\widehat{G}}

\newcommand{\bn}{\mathbf{n}}
\newcommand{\FC}{{\mathscr F}_{\C}}
\def\algorithm{\begin{center}
               \begin{minipage}{6in}
               \begin{tabbing}
               \marks}
\def\falgorithm{\end{tabbing}
                \end{minipage}
                \end{center}}
\def\marks{nn\= nn\= nn\= nn\= nn\= nn\= nn\= \kill}
\theoremstyle{plain}
  \newtheorem{thm}{Theorem}[section]
  \newtheorem{lem}[thm]{Lemma}
  \newtheorem{prop}[thm]{Proposition}
  \newtheorem{cor}[thm]{Corollary}

\theoremstyle{plain}
 \newtheorem{conj}{Conjecture}[section]
\newtheorem{quest}{Question}[section]

\theoremstyle{definition}
  \newtheorem{defn}[thm]{Definition}
  \newtheorem{exmp}{Example}[section]

\theoremstyle{remark}
 \newtheorem{rem}[thm]{Remark}

\numberwithin{equation}{section}
\numberwithin{thm}{section}

\title [Intractability of Hilbert's Nullstellensatz implies hardness of permanent]{Intractability of Hilbert's Nullstellensatz implies\\ 
 algebraic hardness of permanent}
\date{\today}

\author{Peter B\"urgisser}
\thanks{The author was supported by the ERC under the European's Horizon~2020 research and innovation programme (grant agreement no.~787840)} 
\address{Institute of Mathematics, Technische Universit\"at Berlin}
\email{pbuerg@math.tu-berlin.de}

\dedicatory{Dedicated to Brigitte, who accompanied me on this incredible journey during the last thirty years}

\begin{document}

\begin{abstract}
We study the logical relation of the P-NP separation conjecture in the Blum-Shub-Smale-model 
over the complex numbers with the P-NP separation conjecture in Valiant's algebraic model. 
This amounts to comparing Hilbert's Nullstellensatz Problem, that is, 
deciding feasibility of a given system of polynomial equations over the complex numbers,
with the problem of evaluating the permanent of a given complex matrix.
We compare the respective uniform models of computations and prove that 
$\Po_\C\ne \NP_\C$ in the Blum-Shub-Smale-model over~$\C$ 
implies the separation $\VP^0(\U)\ne\VNP^0(\U)$ of the uniform versions of 
Valiant's constant-free complexity classes over~$\C$.
For the nonuniform models we show the analogous implication: the separation 
$\Po^0_\C(\NU)\ne \NP^0_\C(\NU)$ of the nonuniform, constant-free Blum-Shub-Smale classes over~$\C$ 
implies the separation $\VP^0\ne\VNP^0$ of 
Valiant's constant-free complexity classes over~$\C$. 
In the reverse direction, we conjecture that 
$\VNP_\C\not\subseteq \overline{\VP}_\C$ implies that 
$\Po_\C(\NU) \ne \NP_\C(\NU)$.
\end{abstract}



\maketitle


\section{Introduction}

\subsection{Hilbert Nullstellensatz Problem and Blum-Shub-Smale-model over \texorpdfstring{$\C$}{C}}

The \emph{Hilbert Nullstellensatz Problem $\HN_\C$} 
is the most fundamental computational problem of algebraic geometry. 
It is the problem of deciding for given complex polynomials $f_1,\ldots,f_s$ of the degrees $d_1,\ldots,d_s$ in $n$ variables, 
whether they have a common zero $x\in\C^n$, i.e., whether 
$$
 \exists x \in\C^n  \ f_1(x_1,\ldots,x_n)=0,\ldots, f_s(x_1,\ldots,x_n)= 0. 
$$
Blum, Shub, and Smale~\cite{blss:89} introduced a uniform model of algebraic computation (BSS-machines),
which can perform arithmetic operations and tests for zero with complex numbers at unit cost.
Using this model, they defined complexity classes $\Po_\C$ and $\NP_\C$ and 
proved that $\HN_\C$ is $\NP_\C$-complete. This is analogous to the well known $\NP$-completeness of 
the Boolean satisfiability problem, formulated in the model of Turing machines. 
The fundamental conjecture is $\Po_\C \ne \NP_\C$, which is equivalent to $\HN_\C\not\in\Po_\C$.  
(An exposition of this theory can be found in ~\cite{BCSS:98}.) 
Smale~\cite{smale-nc:98, smale-nc:00} considers this conjecture, along with its counterpart $\Po\ne\NP$ 
(which is the analogous conjecture over $\F_2$), as one the most important open problems in mathematics.

The problem $\HN_\C^\Z$ is the restriction of $\HN_\C$ to integer coefficient polynomials~$f_i$,
that are thought to be given by the list of their coefficients in binary encoding. 
The problem $\HN_\C^\Z$ can be studied in the model of Turing machines.
It is clear that $\HN_\C^\Z$ is $\NP$-hard: indeed,  viewing bits as the solutions of 
polynomial equations $x_i^2-x_i=0$, 
one easily sees that Boolean satisfiability reduces to $\HN_\C^Z$.
It is a nontrivial result~\cite{canny-conf:88,renegar-I:92} that $\HN_\C^\Z\in\PSPACE$. 
A result by Koiran~\cite{koiran:96} indicates that  $\HN_\C^\Z$ is located in a much lower 
complexity class: Koiran showed that $\HN_\C^\Z\in\AM$, conditional on the Generalized Riemann Hypothesis.
Recently, in a breakthrough work, 
Andrews, Garg and Schost~\cite{andrews-et-al-HN:26}
unconditionally improved the known upper bound for $\HN_\C^\Z$. 
They showed that $\HN_\C^\Z$ lies in the \emph{counting hierarchy}~$\CH$.

\subsection{Evaluation of the permanent and Valiant's model}

Ten years before the paper by Blum, Shub and Smale~\cite{blss:89}, 
a different algebraic model 
for the complexity classes $\Po$ and $\NP$ 
was introduced by Valiant~\cite{vali:79-3}.
He defined nonuniform algebraic  complexity classes $\VP_\F$ and $\VNP_\F$ 
over any field~$\F$, using the model of arithmetic circuits.  
Since the goal was to capture the evaluation of polynomials over~$\F$, 
no branching is included in this model. The absence of uniformity conditions, 
which requires the concept of the Turing machine, is an abstraction  
which adds to the conciseness and elegance of the model. Uniformity constraints 
are easily added without changing much of the developments. 
The $n$th \emph{permanent polynomial}~$\per_n$ is defined as 
\begin{equation*}\label{eq:def-per}
 \per_n = \sum_{\pi\in S_n} \prod_{i=1}^{n} x_{i,\pi(i)} .
\end{equation*}
Valiant proved that the family $(\per_n)$ is $\VNP_\F$-complete if $\chara \F\ne 2$
and conjectured that $\VP_\F \ne \VNP_\F$ for any field $\F$. 
(The validity of this conjecture only depends on the characteristic of $\F$, 
as shown in~\cite[Thm.~4.3]{buer:00-3}.)  
Valiant also showed that the family of \emph{Hamilton cycle polynomials}
is $\VNP_\F$-complete over any field~$\F$.

For comparing algebraic complexity classes over different fields, and for relating them 
to other models of computation, it turns out to be important to work with a modification of these classes.
The \emph{constant-free complexity classes} $\VP^0$ and $\VNP^0$, 
were introduced and studied by Malod~\cite{malodthesis:03}. 
A family~$(f_n)$ of integer polynomials lies in $\VP^0$ if there is a family $(\cC_n)$ of constant-free 
arithmetic circuits such that the number of variables and the formal degree of $\cC_n$ grow at most 
polynomially in~$n$, and $\cC_n$ computes~$f_n$.
Constant-free means that the circuit may only use $0,1$ as free constants. 
Malod~\cite{malodthesis:03} showed that the family 
of Hamilton cycle polynomials is $\VNP^0$-complete. 
The class $\VNP^0$ arises from $\VP^0$ by allowing exponential summation at the final stage of the computation.
The implication $\VP_\C\ne \VNP_\C \Rightarrow \VP^0\ne \VNP^0$ is trivial, but the converse is unclear. 
Expositions of this theory can be found in~\cite{gath:87-1,buer:00-3,malod:07,buerg-survey:24}. 

\subsection{Comparing different P-NP conjectures} 

It is natural to ask about the logical dependency of the different separation conjectures discussed above. 
Evidence for the separation in the BSS-model comes from the observation that 
the conjecture $\Po_\C\ne\NP_\C$ is implied by 
$\NP\not \subseteq\BPP$ in the Turing model; see~\cite{ckklw:95} and~\cite[Prop.~2.1]{abpm:08}. 
The point here is that transcendental constants used by a BSS-machine lead to the task 
of deciding whether a multivariate integer polynomial, computed by a given arithmetic circuit, vanishes 
identically (polynomial identity testing). 
As it is well known, the latter can be efficiently decided using randomization. 

As for the separation in Valiant's model, it is easy to see that 
$\SP\not \subseteq \Po/\poly$ implies $\VP^0 \ne \VNP^0$.
The reason is that constant-free arithmetic circuits working over $\Z$ 
can be executed over finite fields~$\F_p$ (thus avoiding an exponential growth of bit size),
 and this can be efficiently simulated by Boolean circuits.
If arithmetic circuits may use complex constants, the situation is much less clear. 
In~\cite{buerg:00} it was shown that $\SP\not \subseteq \Po/\poly$ implies $\VP_\C \ne \VNP_\C$, 
assuming the Generalized Riemann Hypothesis. This relies on the same methods as employed in~\cite{koiran:96}. 
The mentioned implications are consistent with our common understanding that  lower bounds in algebraic models 
are easier to prove than in the less structured bit models.

A main source of motivation of the work~\cite{buer:00-3} 
was to clarify the logical relation between the \emph{algebraic P-NP separation conjectures} 
formulated in the BSS-model and Valiant's model, respectively; 
see Figure 1.1 and Section~8 there. 
Of course, these conjectures cannot be compared directly, 
since the BSS-model is uniform and Valiant's model is nonuniform.
As for Valiant's model, it is useful to focus on complexity classes $\Po^0_\C$ and $\NP^0_\C$ 
that are defined via \emph{constant-free BSS-machines}, which may only use $0,1$ as free constants. 
This does not have any effect on the separation conjectures, since 
one can show that 
$\Po_\C\ne \NP_\C$ is equivalent to $\Po^0_\C\ne \NP^0_\C$: 
the proof relies on the possibility to eliminate complex constants using
witness sequences, as developed in~\cite{BCSS:95,BCSS:98,koir:97,koir:99}.

For the sake of comparison, we define the nonuniform versions $\Po^0_\C(\NU)$ and $\NP^0_\C(\NU)$
of these classes in~\cref{se:BSS-classes}. 
We also introduce the uniform versions $\VP^0(\U)$ and $\VNP^0(\U)$ 
of Valiant's constant-free classes in~\cref{se:uniform-Val}. 

\subsection{Main result}\label{se:mainR}

Our main result states that the P-NP separation conjecture in the Blum-Shub-Smale model over $\C$
implies the P-NP separation conjecture in Valiant's model over $\C$, 
both for the uniform and nonuniform models, 
in the constant-free settings. 

\begin{thm}\label{th:MAIN-u}
We have the implications: 
\[
\begin{array}{ccc}
 \Po^0_\C(\NU) \ne \NP^0_\C(\NU) & \Longrightarrow &  \VP^0\ne \VNP^0 \\[1.5ex]
 \Po^0_\C \ne \NP^0_\C                 & \Longrightarrow & \VP^0(\U) \ne \VNP^0(\U) .
\end{array}
\]
\end{thm}

Of course, we prove the contraposition, which says that 
the existence of a (uniform) family of constant-free arithmetic circuits~$\cC_n$ of polynomial size in~$n$, 
such that $\cC_n$ computes the permanent~$\per_n$, 
implies the existence of a (uniform) family 
of constant-free algebraic circuits~$\cD_{n,d_1,\ldots,d_n}$  of polynomial size in~$n,d_1,\ldots,d_s$, 
such that $\cD_{n,d_1,\ldots,d_s}$ decides the Hilbert Nullstellensatz Problem $\HN^\Z_\C$.


In my habilitation thesis~\cite[Conj.~8.2]{buer:00-3}, 
I conjectured the converse of~\cref{th:MAIN-u} 
for the nonuniform models in their original versions (any complex constants).
The following conjecture involving border complexity may be easer to resolve, 
since reductions involving border complexity allow for considerably more flexibility, 
as illustrated in~\cite{andrews-forbes:22}. 

\begin{conj}\label{conj:reverse-dir-nu}  
If $\VNP_\C\not\subseteq \overline{\VP}_\C$, then $\Po_\C(\NU) \ne \NP_\C(\NU)$.
\end{conj}

Let us motivate and define the \emph{border complexity class} $\overline{\VP}_\C$. 
Border complexity arises naturally when, besides arithmetic operations, 
one also allows limit processes.
In this model, one can pass at no cost to the 
\emph{trailing coefficient}  
$\tc(f) := f_e\ne 0$ of an already computed multivariate polynomial 
$f= f_e t^e + f_{e+1} t^{e+1} + \ldots$ with respect to the (deformation) variable~$t$. 
Such limit process occurs often in algebraic geometry reasonings: 
we use it crucially for a key reduction in the proof of~\cref{th:MAIN-u}; 
see~\cref{se:deform} and~\cref{se:RED-MAIN}.
In fact, almost all of the known techniques for proving complexity lower bounds actually provide more, 
namely lower bounds for border complexity.

The complexity class $\overline{\VP}_\C$ arises naturally from $\VP_\C$ when replacing 
complexity by approximate complexity. 
The conjecture $\VNP_\C\not\subseteq\overline{\VP}_\C$  
first appeared in~\cite[Conj.~1]{buerg-survey:01} and was investigated in~\cite{buerg:04}. 
Proving this conjecture is the main goal of geometric complexity theory. 
Variants of the class~$\overline{\VP}_\C$  
currently are intensely studied for restricted models of 
computation~\cite{bringmann-et-al:18,kumar:20,dutta-saxena:22a,dutta-saxena:22b,andrews-forbes:22}. 
We refer to~\cite[\S4.7]{buerg-survey:24} for more explanations and references. 

A step toward~\cref{conj:reverse-dir-nu} has recently been achieved~\cite{bcjs:26}: 
by an ingenious reduction it was shown that the existence of a polynomially size family 
of arithmetic circuits for the multivariate resultant~$\cR_n$ of systems of quadratic forms 
implies $\VP_\C=\VNP_\C$.
\cref{conj:reverse-dir-nu} would follow if one could extend this implication 
to nonzero multiples of $\cR_n$. It is plausible that border complexity 
will enter for this; compare~\cite{buerg:04}.

\subsection{Proof techniques }

In~\cite{abpm:08} by Allender, B\"urgisser, Kjeldgard-Pedersen, and Bro Miltersen,  
the counting hierarchy was identified as a helpful concept 
for understanding the power of arithmetic
circuits (compare the survey~\cite{allender-survey:04}).
In~\cite{abpm:08} it was shown that languages decided by constant-free polynomial time real BSS-machines on Boolean inputs
lie in the \emph{counting hierarchy} $\CH$, a class believed to be strictly contained in $\PSPACE$.  
In~\cite{buerg:09} these ideas were used to connect some so far unrelated 
conjectures in algebraic complexity theory. A key insight was that certain 
\emph{specific sequences of integers} of exponential bitsize, like factorials and binomials, 
can be \emph{defined in the counting hierarchy.}   
These ideas (along with boosting Valiant's criterion to polynomials of exponential degree) 
were taken up by Koiran and Perifel~\cite{koir-per:11}. 
They implicitly defined a complexity class---we shall call it the \emph{algebraic counting complexity class $\VCH^0$}---which 
turns out to be closed under many desirable operations. 
This class is contained in $\VPSPACE^0$, the arithmetic analogue of $\PSPACE$, 
which was introduced by Poizat~\cite{poizat:08} and Koiran and P\'erifel~\cite{koir-perifel-R}. 
The replacement of $\VPSPACE^0$ by the smaller class $\VCH^0$ is the critical new step here. 
We develop the theory of $\VCH^0$ in \cref{se:VCH}. 

The missing link for completing the puzzle and establishing~\cref{th:MAIN-u} was provided in the recent breakthrough paper 
by Andrews, Garg and Schost~\cite{andrews-et-al-HN:26}. 
The critical ingredient of the proof is an explicit formula for the power series expansion of implicit functions,
which is due to Ju{\v{z}}akov~\cite{juzakov:75}. 
The work in~\cite{andrews-et-al-HN:26}  implies that families of \emph{multivariate resultants}  
of systems of quadratic forms lie in $\VCH^0$. 
Moreover, \cite{andrews-et-al-HN:26} uses an ingenious reduction from the Hilbert Nullstellensatz Problem
to the multivariate resultant, 
which is made possible by Canny's~\cite{canny-gcp:89} concept of the generalized characteristic polynomial. 
(Related important papers are \cite{canny-conf:88} and \cite{renegar-I:92}.)
In~\cite{andrews-et-al-HN:26} this reduction is used as elaborated by Ierardi~\cite{ierardi:89}. 
We adjust this reduction to employ it in the BSS-model. 
This can be done in an elegant way by using Koiran's concept of witness sequences~\cite{koiran:97}, 
which allows for a general derandomization argument.

\subsection{Outline of paper}

In \cref{se:prelim} we first define the nonuniform Blum-Shub-Smale complexity classes 
and uniform Valiant complexity classes, which appear in the main result.
Then we recall Koiran's notion of generic quantifiers and witness sequences. 
We also present the key concept of sequences definable in the counting 
hierarchy from B\"urgisser~\cite{buerg:09} along with its properties.
\cref{se:resultant} recalls the basics about multivariate resultants and 
describes a deformation technique following Zulehner~\cite{zulehner:88}.

In~\cref{se:VCH}, the new algebraic counting complexity class $\VCH^0$ is defined 
and its properties are investigated. Many of the ideas appear already in 
Koiran and Perifel~\cite{koir-per:11}. This section should be of independent interest. 
\cref{se:redu} contains the proof of \cref{th:MAIN-u}, our main result.
We end in~\cref{se:outlook} with an outlook proposing several open questions.

\section{Concepts and tools from complexity theory}\label{se:prelim}

\subsection{Nonuniform BSS-classes}\label{se:BSS-classes}

It is useful to introduce the complexity classes $\Po^0_\C$ and $\NP^0_\C$ 
that are defined via constant-free BSS-machines, which may only use $0,1$ as free constants. 
The completeness proof for $\HN_\C$ immediately shows that this problem remains complete for 
the smaller class $\NP^0_\C$ for constant-free polynomial time reductions. 
Moreover, the conjecture $\Po_\C= \NP_\C$ is equivalent to $\Po^0_\C= \NP^0_\C$, 
which relies on the possibility to eliminate complex constants using
witness sequences, as developed in~\cite{BCSS:98,koir:97,koir:99}.

As  pointed out in~\cite{koiran:00,poizat:95}, the classes $\Po_\C$ and $\NP_\C$ 
can be equivalently defined in terms of algebraic circuits, 
without using the notion of BSS-machines.
It is convenient to use this approach for defining the nonuniform versions 
of the  complexity classes $\Po^0_\C$ and $\NP^0_\C$. 

An \emph{algebraic circuit} has computation gates for arithmetic operations and selector gates, which allow to branch according 
to whether a computed number equals zero or not (see~\cite{gath:86}; the details of this model are not important 
in our context). However, we assume that the circuit is constant-free, which means that $0$ and $1$ are the only free constant 
the circuits may use. 
An algebraic circuit $\cD$ with $n$ input nodes models the computation of a piecewise rational function $\C^n\to\C$.
(On each constructible piece, the function is given by the quotient of the two polynomials with integer coefficients; 
compare \cite[\S 4.4-4.5]{ACT}).
The algebraic circuit can also model the computation of a function $\C^n\to\{0,1\}$ when using a selector gate as output gate
and thus decide membership to a (constructible) subset of $\C^n$. 
We define the complexity class $\Po^0_\C(\NU)$ as the set of problems decided by families $(\cD_n)$ 
of algebraic circuits of size polynomially bounded in~$n$. 
Moreover, we define the complexity class $\NP^0_\C(\NU)$ as the set of decision problems $A$ 
for which there exists $B \in\Po^0_\C(\NU)$ such that 
$x\in A \Longleftrightarrow \exists y\in\C^{p(n)}\ (x,y)\in B$, 
where $p(n)$ is polynomial in in $n$. From the usual argument, it follows that 
$\HN_\C$ is $\NP^0_\C(\NU)$-complete with respect to nonuniform constant-free polynomial time reductions.

\subsection{Uniform Valiant's classes}\label{se:uniform-Val}

We introduce here uniform versions of 
Valiant's complexity classes in the obvious way. 
The class $\VP^0(\U)$ consists of the families $(g_n) \in \VP^0$ 
such that there is a family $(\cC_n)$ of constant-free, polynomial size arithmetic circuits,  
computed by a Turing machine in time polynomial in $n$, 
and such that $\cC_n$ computes~$g_n$.  
The class $\VNP^0(\U)$ consists of the families~$(f_n)$, where $f_n$ results from $g_n$
by an exponential summation (cf.~\cite[Def.~2.25]{buerg-survey:24})
$$
 f_n(x_1,\ldots,x_{v(n)}) = \sum_{e\in \{0,1\}^{u(n)-v(n)}} g_{n}
 (x_1,\ldots,x_{v(n)},e_{v(n)+1},\ldots,e_{u(n)})  .
$$
We note that the family of Hamilton cycle polynomials is $\VNP^0(\U)$-complete, which follows 
by checking the proof in~\cite{malodthesis:03}.   

\subsection{Generic quantifiers and witness sequences}\label{se:witness-s}

Reasonings in algebraic geometry heavily rely on the notion of a generic point, whose 
precise meaning has changed over time. Originally, it meant a point whose 
components are algebraic independent over a coefficient field. 
Many reductions in computational algebraic geometry rely on this notion: 
e.g., to intersect or project to a subspace in general position.
This is usually realized by allowing random coins, which leads to 
randomized algorithms. 

It is well known that by repeated squaring, polynomial size arithmetic circuits can 
produce sequences of integers of exponential size that in some sense are 
indistinguishable from algebraic independent elements over $\Q$.
Strassen's paper~\cite{stra:74-1} appears to be the earliest appearance of this idea in complexity theory; 
see also~\cite[Lemma~(9.30)]{ACT}. 
This idea was used in Heintz and Schnorr~\cite{hesc:82} for constructing correct test sequences, 
which gave insights into the polynomial identity testing problem. 

Koiran~\cite{koiran:97,koir:99,koir:99b} proposed the notion of generic quantifiers in the 
framework of algebraic computations and, based on the above idea,
developed a general method for eliminating generic quantifiers in parametrized formulas.
This leads to a general derandomization argument. 
This way, Koiran~\cite{koiran:97} proved that the problem of deciding whether a given affine algebraic variety 
has dimension at least~$d$ is $\NP_\C$-complete (the difficulty is to establish the upper bound 
without resorting to randomization). 
This was further extended in~\cite{bucu-count-II:06} in order to deal with counting problems 
over~$\C$ and $\R$. 

We recall here the main features of the above method, which will allow to give 
a short and clean analysis of the reduction in \cref{se:redu}. 
More details can be found in~\cite[\S7.1]{bucu-quaderni:04}. 

Let us denote by $\FC$ be the set of first order formulas over the language
of the theory of fields with constant symbols for complex numbers.

\begin{defn}
Let $F\in\FC$ have free variables $a_1,\ldots,a_q$.
We say that $F$ is {\em Zariski-generically true}
if the Zariski-closure of the set of values $a\in\C^q$ not satisfying $F(a)$
has dimension strictly less than~$q$.
We express this fact by writing $\forall^\ast a\, F(a)$
using the {\em generic universal quantifier} $\forall^\ast$.
\end{defn}

It is not hard to see that 
$$
 \forall^\ast a\, \forall^\ast b\ F(a,b) \equiv \forall^\ast (a,b) F(a,b) .
$$
Let $K$ denote the field generated by the coefficients of all 
the polynomials occuring in $F$.
One can show that 
$\forall^\ast a\, F(a)$ is equivalent to requiring that 
$F(\a)$ holds whenever the components 
$\a_1,\ldots,\a_q$ of~$\a$ are algebraically independent over $K$.
Thus we may view $\alpha$ as a witness for the fact $\forall^\ast a\, F(a)$.

We focus now on parametrized formulas $F(u,a)$ with parameter $u\in\C^p$
and look for witnesses which can be used for all values of the
parameter~$u$. This may not be attainable with a single witness point,
but it turns out to be doable by using short sequences of witness 
points and taking a majority vote.
In the sequel, $[n]$ denotes the set $\{1,\ldots,n\}$.

\begin{defn}\label{def:pws}
Let $F(u,a)\in\FC$ with variables
$u\in\C^p$ and $a\in\C^q$. We call a sequence
$\alpha=(\alpha_1,\ldots,\alpha_{2p+1})\in(\C^q)^{2p+1}$
a {\em witness sequence} for~$F$ iff
$$
 \forall u\in\C^p\ \Big(
   \forall^\ast a\in\C^q\  F(u,a) \Longleftrightarrow
  \#\{\theta\in [2p+1] \mid F(u,\alpha_\theta)\} > p \Big) .
$$
We denote the set of witness sequences of $F$ by $W_\C(F)$.
\end{defn}

By a transcendence degree argument one can show that 
$W_\C(F)$ is Zariski dense in $\C^{q(2p+1)}$,
for any $F(u,a)\in\FC$; 
see~\cite[Thm.~5.1]{koir:97}.

Suppose we have polynomials 
$P_1,\ldots,P_M \in\Z[u_1,\ldots,u_p, a_1,\ldots,a_q]$ 
in two group of variables 
with degree bounded by $d$ and bit size bounded by $\ell$.
They define the atomic predicates 
$\sgn(P_i(u,a))$, 
where we set 
$\sgn(c):= 0$ if $c=0$ and $\sgn(c):= 1$ if $c\ne 0$. 
Using a Boolean function
$B\colon\{0,1\}^M \to \{0,1\}$,  
we can combine the atomic predicates to the formula 
\begin{equation}\label{eq:wf-formula}
 B(\sgn(P_1(u,a)),\ldots, \sgn(P_M(u,a)) = 1 . 
\end{equation}
The next theorem is from~\cite{koiran:97}.

\begin{thm}\label{th:cons-pws}
Let $F(u,a)\in\FC$ be a quantifier-free formula as in~\eqref{eq:wf-formula}
with free variables $u\in\C^p$ and $a\in \C^q$, 
with $M$ atomic predicates given by the polynomials 
$P_1,\ldots,P_M \in\Z[u_1,\ldots,u_p, a_1,\ldots,a_q]$ 
of degree at most~$d$ and 
with integer coefficients of bit size at most~$\ell$.
\begin{enumerate}
\item Then a witness sequence
$\alpha\in W_\C(F)\cap (\Z^q)^{2p+1}$ can be computed
by an arithmetic circuit~$\Gamma$ of size 
$(pq)^{O(1)}\, \log (M d) + O(\log\ell)$,
which is division-free, has~$1$ as its only constant and no inputs.

\item There exists a Turing machine which, on input $s\in\N$ in unary,  
such that $p,q,L\le s$ and $\log\max\{M, d,\ell\} \le s$, 
computes~$\Gamma$ in time polynomial in~$s$.
This machine does not depend on~$F$.
\end{enumerate}
\end{thm}

\begin{proof}
The first part is~\cite[Thm.~5.6]{koiran:97}
in the special case of quantifier-free formulas.
For the second part concerned with the Turing machine, 
the reader should inspect the construction of the sequence 
in the proof of~\cite[Lemma~5.4]{koiran:97}, which basically 
proceeds by repeated squaring.
\end{proof}

\subsection{Sequences of integers definable in the counting hierarchy}\label{se:int-def-CH}

The title names a concept introduced in~\cite{buerg:09}, 
which is of central importance for our developments. 
It involves the counting hierarchy, which was introduced by Wagner~\cite{wagn:86}.
The counting hierarchy is closely tied to the theory of threshold circuits of bounded depth; see~\cite{alwa:93}.

We recall that the classes of the polynomial hierarchy are obtained from the class~$\Po$
by iteratively applying the operators $\exists$ and $\forall$: 
the union of the resulting classes is denoted $\PH$. 
The classes of the counting hierarchy are obtained analogously by repeatedly 
applying a counting operator to the class~$\Po$. The first class obtained this way is $\PP$
and the union of the classes is denoted $\CH$; 
we refer to~\cite{alwa:93,buerg:09} for more details. We have the inclusions
$$
 \PH \subseteq \CH \subseteq \PSPACE .
$$
The collapse $\Po=\PP$ implies that $\Po=\CH$, similarly as 
$\Po=\NP$ implies $\Po=\PH$.  
This observation is important for our proof.  

We begin by introducing a convenient notion (compare \cref{def:exp-format}).

\begin{defn}
A \emph{family of integers of exponential bitsize} is a family
$(c_n(k_1,\ldots,k_{u(n)}))$ of integers, where $u(n)$ is a polynomial, 
and such that there is a polynomial~$p(n)$ 
with the property that 
$c_n(k_1,\ldots,k_{u(n)})$ is defined for 
$n,k\in\N$ with $0\le k_i < 2^{p(n)}$ and 
$c_n(k_1,\ldots,k_{u(n)})$ has bitsize at most~$2^{p(n)}$. 
\end{defn}

To ease notation, we will focus in the following on families $(c_n(k))$ where $u(n)=1$. 
We shall think of the integer $k$ as being represented in binary encoding of length at most~ $2^{p(n)}$. 
For example, the family of factorials $c_n(k) = k!$, where $0< k \le 2^{n}$,
defines a family of integers of exponential bitsize. 

We associate with a family  $c=(c_n(k))$ of integers of exponential bitsize the language
\begin{equation*} 
\Bit(c)  := \{ (1^n, k, a) \mid\mbox{ the $a$-th bit of $c_n(k)$ equals $1$ } \} ;
\end{equation*}
the integers~$k,a$ are thought to be represented in binary using $p(n)$ bits.   
Throughout, we shall use the convention that the $0$-th bit of $a$ gives the sign of  $c_n(k)$.
Note that $a \le p(n)$ since we assume $\log c_n(k) \le 2^{p(n)}$.

The following definition is from~\cite{buerg:09}. 
We present the concept in the slightly modified form used by~\cite{koir-per:11}.
(The difference is to use a unary encoding for $n$ instead of a binary.)

\begin{defn}
Let $c$ be a family of integers of exponential bitsize. 
We call~$c$ \emph{definable} in $\CH$ iff their associated language $\Bit(c)$ is in $\CH$. 
Moreover, we call $c$ \emph{definable} in $\CH/\poly$  iff $\Bit(c)\in\CH/\poly$. 
\end{defn}


The $tau$-{\em complexity}~$\tau(a)$ of an integer~$a$
is defined as the minimal size of a division-free arithmetic circuit computing~$a$ 
from $1$. Using addition chains, it is easy to see that 
$\tau(a)\le 2\log_2 a$ if $a$ is a positive integer. 
The following result~\cite[Cor.~3.9]{buerg:09} provides first interesting examples 
of families of integers that are definable in $\CH/\poly$.

\begin{prop}\label{pro:slp-2-CH}
Suppose $(c_n(k))$ is a family of integers of exponential bitsize, where 
$0\le k < 2^{p(n)}$ for a polynomial $p(n)$, such that 
$\tau(c_n(k)) \le p(n)$.
Then $(c_n(k))$ is definable in $\CH/\poly$.
\end{prop}

An example for such a family is given by $c_n(k) = 2^k$. 
It worthwile to directly check that 
$\Bit(c_n(k))$ lies in $\Po$. In particular, 
$(c_n(k))$ lies in $\CH$. 

The following result is~\cite[Thm.~3.10]{buerg:09}.
It expresses that definability in the (nonuniform) counting hierarchy is preserved 
under exponential iterated addition, iterated multiplication, 
as well as under integer division.

\begin{thm}\label{th:closure-CH}
\begin{enumerate}
\item  Suppose $c=(c_n(k))$ is definable in $\CH$, 
where $0\le k < 2^{p(n)}$ with a polynomial~$p$. Consider
$$
 a_n := \sum_{k=0}^{2^{p(n)}-1} c_n(k),\quad
 b_n := \prod_{k=0}^{2^{p(n)}-1} c_n(k).
$$
Then $a=(a_n)$ and $b=(b_n)$ are definable in $\CH$.
Moreover, if $c$ is is definable in $\CH/\poly$, then so are $a$ and $b$.

\item Suppose $(a_n)$ and $(b_n)$ are definable in $\CH$ and $b_n>0$ for all~$n$.
Then the family of quotients
$(\lfloor a_n/b_n \rfloor)$ is definable in $\CH$.
The analogous assertion holds for $\CH/\poly$.
\end{enumerate}
\end{thm}

For instance, \cref{th:closure-CH} implies that the factorials 
$a_n(k) := k!$ for $k\le 2^n$ form a family definable in $\CH$. 
(This is~\cite[Cor~3.11]{buerg:09}.) 

\begin{rem}\label{re:clos-SP-CH}
\cref{th:closure-CH} also holds for families 
with a polynomial number of parameters. 
For instance, suppose $(c_n(k,\ell))$ is definable in $\CH$, 
where $0\le k,\ell < 2^{p(n)}$ with a polynomial~$p$. 
Then 
$a_n(\ell):= \sum_{k=0}^{2^{p(n)}-1} c_n(k,\ell)$ defines 
a family $(a_n(\ell))\in\CH$, and similarly for taking products.
\end{rem}

\section{Multivariate resultants}\label{se:resultant}

The theory of resultants is a classical part of elimination theory, which aims 
at solving systems of polynomial equations. 
Accessible accounts of this theory are
van der Waerden~\cite[\S 82]{waer:48},
Lang~\cite[Chap.~IX]{lang:93} and 
Cox, Little, O'Shea~\cite{cox-et-al-using-AG:98}. 
More advanced treatments can be found in 
the paper~\cite{jouanolou:91} by Jouanolou 
and in the monograph~\cite{GKZ:94} by Gelfand, Kapranov and Zelevinsky. 

We fix natural numbers $d_0,\ldots,d_n$ and
consider homogeneous polynomials 
$F_i\in\C[x_0,\ldots,x_n]$ of degree~$d_i$.
The corresponding zero set in the complex projective space~$\proj^{n}$ 
will be denoted by $V(F_0,\ldots,F_n)$. 
We shall denote by~$u$ a vector of $p=\sum_{i=0}^n \binom{n+d_i }{n}$ variables,
which when specialized to an element of $\C^N$, defines a sequence $(F_0,\ldots,F_n)$. 

The {\em (multivariate) resultant} $\Res_{d_0,\ldots,d_n}\in\Z[u]$ is characterized as the unique 
absolutely irreducible polynomial with the following property:
\begin{enumerate}
\item Any system $F_0(x)=0,\ldots,F_{n}(x)=0$ of homogeneous polynomials of 
the degrees $d_0,\ldots,d_n$ has a nontrivial solution in $\proj^{n}$ if and only if 
$\Res_{d_0,\ldots,d_n}(F_0,\ldots,F_n)=0$.
\item $\Res_{d_0,\ldots,d_n}(x_0^{d_0},\ldots,x_n^{d_n})=1$.
\end{enumerate}
It is known that $\Res_{d_0,\ldots,d_n}$ is multihomogeneous: for each~$i$ it is of 
degree $d_0\cdots d_n/d_i$ in the variables of~$F_i$. 
The resultant is a semi-invariant with respect to linear transformations: 
for $g\in\GL(n+1)$, we have 
$$
 \Res_{d_0,\ldots,d_n}(F_0(gx),\ldots,F_n(gx)) = \det(g)^{d_0\cdots d_n}  \Res_{d_0,\ldots,d_n}(F_0,\ldots,F_n) .
$$
Resultants are are closely related to determinants of highly structured matrices of exponential format: 
Macaulay expressed $\Res_{d_0,\ldots,d_n}$ as a quotient of two determinants,
see~\cite[Chap.~3, Thm~4.9]{cox-et-al-using-AG:98}. 

One can show that the bitsize of the coefficients of $\Res_{d_0,\ldots,d_n}$ 
is bounded by $O(n d^n \log d)$: this follows e.g., from~\cite{sombra:04}, 
which contains a bound in a much more general situation.




The following is a consequence of the Poisson formula, 
which allows to recursively compute the resultant: 
it is a key computational tool used in~\cite{andrews-et-al-HN:26}.

\begin{prop}\label{prop:cor-poisson}
Let $F_1,\ldots,F_{n}\in\C[x_0,\ldots,x_n]$ be homogeneous 
of the degrees $d_1,\ldots,d_n$ such that 
$V(F_1,\ldots,F_{n}) \cap V(x_0) =\varnothing$. Then: 
\begin{enumerate}
\item $V(F_1,\ldots,F_{n})$ is finite. 
We denote by $m(\xi)$ the multiplicity of the zero $\xi$. 

\item There exists $c\in\C^*$ such that for any $H\in\C[x_0,\ldots,x_n ]$ homogeneous of degree~$d_0$, 
$$
 \Res_{d_0,d_1,\ldots,d_{n}}(H,F_1,\ldots,F_n) = 
  c\cdot \prod_{\xi\in V(F_1,\ldots,F_n)} \Big(\frac{H(\xi)}{\xi_0^{d_0}}\Big)^{m(\xi)} .
$$
\end{enumerate}
\end{prop}

\begin{proof}
(1) If $\dim V(F_1,\ldots,F_{n})\ge 1$, then, 
by basic dimension theory of projective varieties, 
the zero set would intersect the hyperplane $V(x_0)$,  
contradicting our assumption.

(2) Set $\overline{F}_i(x) := F_i(0,x_1,\ldots,x_{n})$. 
By our assumption, 
$V(\overline{F}_1,\ldots,\overline{F}_{n}) = \varnothing$
and the above property~(1) of the resultant implies 
$\Res_{d_1,\ldots,d_{n}}(\overline{F}_1,\ldots,\overline{F}_{n})\ne 0$.
Now we use the \emph{Poisson formula}
\begin{equation*}
 \Res_{d_0,\ldots,d_n}(F_0,\ldots,F_n) = (-1)^{s(d_0.\ldots,d_n)} \, \Res_{d_1,\ldots,d_{n}}(\overline{F}_1,\ldots,\overline{F}_{n})^{d_0} \cdot
  \prod_{\xi\in V(F_1,\ldots,F_{n})} \Big(\frac{F_0(\xi)}{\xi_0^{d_0}}\Big)^{m(\xi)} , 
\end{equation*}
where $s$ denotes a function determining the sign. 
This statement appears in~\cite[\S 3.3, Exercise~8]{cox-et-al-using-AG:98};
a complete proof can be found in~\cite{jouanolou:91}. 
\end{proof}

If we take $H = \sum_{i=0}^n u_i x_i$, then 
$\Res_{1,d_1,\ldots,d_n}(F_0,\ldots,F_n) \in \Z[u_0,\ldots,u(n)]$
is called the \emph{$u$-resultant}. According to~\cref{prop:cor-poisson}
the $u$-resultant factors into a product of linear forms, 
from which the finitely many zeros $\xi\in V(F_1,\ldots,F_{n})$, 
along with their multiplicity~$m(\xi)$, can be read off.

\subsection{Deformation for extracting isolated solutions}\label{se:deform}

When homogenizing polynomials, spurious components may arise at infinity. 

\begin{exmp} 
Take $g_1=x_1, g_2 =x_2 + x_1 x_3, g_3= x_3 + x_1 x_2$. 
Then $V(g_1,g_2,g_3)=\{(0,0,0))\}$, but $V(G_1,G_2)$ contains the line 
$V(x_0,x_1)\simeq\proj^1$ at infinity.
\end{exmp}

In order to cope with this phenomenon, we deform the polynomials,
which is a standard idea in intersection theory, e.g, 
see~\cite[\S 6.4.2]{eisenbud-harris:16}. 

Suppose $g_1,\ldots,g_n$ are polynomials in $\C[x_1,\ldots,x_n]$ 
with finite zero set $V(g_1,\ldots,g_n)\subseteq\C^n$. 
Set $d_i=\deg g_i$ and consider the homo\-genization of ~$g_i$, 
$$
 G_{i} := x_0^{d_i} g_i(x_1/x_0,\ldots,x_n/x_0) .
$$  
Then $V(g_1,\ldots,g_n)=V(G_1,\ldots,G_n) \cap \{x_0\ne 0 \}$. 
Clearly, the points in $V(g_1,\ldots,g_n)$ are isolated in $V(G_1,\ldots,G_n)$ 
since $V(g_1,\ldots,g_n)$ is finite. 
We introduce a new variable~$t$ 
(interpreted as a deformation parameter) and define  
\begin{equation}\label{eq:def-hG}
 \hG_i := G_{i} + t (x_i^{d_i} - x_0^{d_i}), \quad\mbox{$i=1,\ldots,n$} ,
\end{equation}
which is homogenous of degree~$d_i$ in the $x$-variables. 
It is important that 
\begin{equation}\label{eq:hG-no-pts-at-infty}
 V(\hG_1,\ldots,\hG_n)\cap V(x_0) = \varnothing ,
\end{equation}
so that $\hG_1,\ldots,\hG_n$ satisfies 
the assumption of~\cref{prop:cor-poisson}. 
This follows from 
$$
 \Res_{d_1,\ldots,d_n}(\widetilde{G}_1,\ldots,\widetilde{G}_n) \ne 0 ,
$$
where $\widetilde{G}_i= g_{d_i} + t x_i^{d_i}$ denotes the dehomogenization of $\hG_i$ 
and $g_{d_i}$ is the degree~$d_i$ homogeneous part of~$g_i$.
The nonvanishing of the above resultant is a consequence of~\cite[Thm.~3.8]{lang:93}.

We consider the zero set 
$$
 Z:= V(\hG_1,\ldots, \hG_n) \subseteq \proj^n\times \C
$$
with the fibers $Z_\tau \subseteq \proj^n$ defined by 
$\pi^{-1}(\tau) = Z_\tau \times \{ \tau\}$, 
where $\pi\colon \proj^n\times \C \to \C$ 
denotes the projection. 
Clearly, $Z_0 = V(G_1,\ldots,G_n)$.  
The proof of the following result is from~\cite{zulehner:88}: 
it requires some basic algebraic geometry and 
the implicit function theorem. 
(We note that \cite{ierardi:89} gives a proof over any algebraically closed field
for perturbations $G_{i} + t x_i^{d_i}$.) 

\begin{prop}[Zulehner]\label{pro:curves+limits}
Let $D:=d_1\cdots d_n$.
\begin{enumerate}
\item There is a finite subset $E\subseteq \C$ containing $0$ such that 
$Z\cap (\proj^n \times (\C\setminus E))$ is a union of 
disjoint smooth algebraic curves $\Gamma_1,\ldots, \Gamma_D$ 
parameterized by analytic functions 
$p^{(1)} \colon \C\setminus E \to \proj^n$, $j=1,2,\ldots,D$, so that 
$Z_\tau = \{p^{(1)}(\tau),\ldots,p^{(D)}(\tau) \}$ for $\tau\notin E$.

\item For each $j$, the limit
$q^{(j)} := \lim_{\tau\to 0} p^{(j)}(\tau)$
exists and lies in $V(G_1,\ldots,G_n)$. 

\item
Every isolated zero of  $V(G_1,\ldots,G_n)$ occurs among the limit points~$q^{(j)}$. 
Moreover, $V(g_1,\ldots,g_n)$ exactly consists of the affine limit points $q^{(j)}\in\C^n$.
\end{enumerate}
\end{prop}

Let us assume that 
$q^{(1)},\ldots, q^{(r)}$ are the affine limit points in $\C^n$: 
for those we will assume that 
$q^{(j)}_0=1$ (without loss of generality).  
We next explain how to retrieve information about 
the affine limit points from 
the trailing coefficient of the resultant of the deformed 
$\hG_1,\ldots,\hG_n$ and a freely chosen homogeneous form~$H$.
The definition of trailing coefficients is given in~\eqref{eq:TC}.

\begin{prop}\label{pro:TC-RES}
Let $H\in\C[x_0,\ldots,x_n]$ be homogeneous of degree~$d_0$ and 
$\hG_1,\ldots, \hG_n \in \C[t][x_0,\ldots,x_n]$ be as
in~\eqref{eq:def-hG}.
Let $q^{(1)},\ldots, q^{(r)}$ denote the affine limit points of 
the zero set
$Z= V(\hG_1,\ldots, \hG_n)$ in $\proj^n\times \C$
as defined in \cref{pro:curves+limits}.
If we assume that $\prod_{j=1}^r H(q^{(j)})$ is nonzero, then 
the trailing coefficient $\tc_t(R)$ with respect to the variable~$t$ of the resultant 
$$
 R(t) :=\Res_{d_0,d_1,\ldots,d_{n}}(H,\hG_1,\ldots,\hG_n) \in \C[t] , 
$$
factors as ($c'\in\C^*$) 
$$
 \tc_t(R) = c' \prod_{j=1}^r H(q^{(j)})  .
$$
\end{prop}

\begin{proof}
By~\cref{pro:curves+limits}, 
when specializing $t$ to $\tau\in\C\setminus E$, 
the homogeneous polynomials $\hG_1,\ldots,\hG_n$ in $x_0,\ldots,x_n$ have 
exactly $D=d_1\cdots d_n$ many zeros in $\proj^n$,
hence the multiplicities of the zeros equal one by B\'ezout's theorem. 
We can apply~\cref{prop:cor-poisson} since~\eqref{eq:hG-no-pts-at-infty} holds.
Hence there exists $c(\tau)\in\C^*$ (not depending on $H$) such that
$$
 R(\tau) = \Res_{d_0,d_1,\ldots,d_{n}}(H,\hG_1,\ldots,\hG_n) 
   = \prod_{j=1}^r \frac{H(p^{(j)}(\tau))}{(p^{(j)}(\tau)_0)^{d_0}} \cdot 
   c(\tau) \prod_{j=r+1}^D \frac{H(p^{(j)}(\tau))}{(p^{(j)}(\tau)_0)^{d_0}} .
$$
Here the $p^{(1)},\ldots, p^{(r)}$ parametrize the curves with affine limit points $q^{(1)},\ldots, q^{(r)}$.
The first factor, expanded into a power series in $\tau$, gives (recall $q^{(j)}_0=1$) 
$$
 \prod_{j=1}^r \frac{H(p^{(j)}(\tau))}{(p^{(j)}(\tau)_0)^{d_0}} 
= \prod_{j=1}^r H(q^{(j)})  + O(\tau) , \mbox{ for $\tau\to 0$} .
$$
The constant term is nonzero by our assumption. 
We write the second factor in the following form, for some $c'\in\C^*$ and $e\in\Z$, 
$$
 c(\tau) \prod_{j=r+1}^D \frac{H(p^{(j)}(\tau))}{(p^{(j)}(\tau)_0)^{d_0}} 
    = c' \tau^e + O(\tau^{e+1} ) \mbox{ for $\tau\to 0$} .
$$
The integer $e$ must be nonnegative since $R(\tau)$ depends polynomially on $\tau$. 
This proves the assertion. 
\end{proof}

\begin{rem}\label{re:order-e}
The order~$e$ of vanishing appearing in the proof 
encodes information about the components of $V(G_1,\ldots,G_n)$ of positive dimension,  
as well as about the speed by which the solutions $p^{(j)}$ go to infinity as $\tau\to 0$,  for $j>r$. 
We believe that $e$ may grow exponentially fast in $n$. 
\end{rem}

\begin{rem}
The idea behind the proof is the source of the powerful homotopy methods 
for numerically solving systems of polynomial equations. 
E.g., see~\cite{zulehner:88} and the references given there.
A complexity analysis of homotopy methods can be found in~\cite[Part III]{condition}.
\end{rem}

\section{The complexity class \texorpdfstring{$\VCH^0$}{VCH0}}\label{se:VCH}

The multivariate resultants are at the heart of the Hilbert Nullstellensatz Problem.
Since the degree of these resultants grows exponentially in the number of variables (see \cref{se:resultant}), 
their complexity cannot be described by the classes $\VP^0$ or $\VNP^0$. 
We capture this by the following convenient definition. 

\begin{defn}\label{def:exp-format}
A family $(f_n)$ of multivariate polynomials over~$\Z$ 
is said to be of \emph{exponential format} 
if there is a polynomial $p(n)$ satisfying:
\begin{enumerate}
\item the number of variables of $f_n$ is bounded by $p(n)$,
\item the degree and the bitsize of the coefficients of $f_n$ are bounded by $2^{p(n)}$. 
\end{enumerate}
\end{defn}

Malod~\cite{malodthesis:03} defined the complexity class $\VPnb^0$ similarly as $\VP^0$, 
but without requiring that the degree of~$f_n$ is polynomially bounded.
For instance, repeated squaring defines the family $f_n:= x^{2^n}$, which is in~$\VPnb^0$. 
Clearly, every $(f_n)\in\VPnb^0$ has exponential format.

By exponential summation one defines the class $\VNPnb^0$, which results from $\VPnb^0$ 
in the same way as $\VNP^0$ results from $\VP^0$; see \cite{malodthesis:03} and \cite[\S 4.2]{buerg-survey:24}.
Clearly, every $(f_n)\in\VPnb^0$ has exponential format.

An important property of $\VNP^0$ is that it is closed  
with respect to passing to coefficient sequences~\cite{malodthesis:03}.  
However, for $\VNPnb^0$, this is only known in positive characteristic~\cite{malod:07}; 
see also~\cite[\S 4.2]{buerg-survey:24}. 
Moreover, it is unknown whether the sequence of multivariate resultants lies in $\VNPnb^0$; see~\cite{grenet-et-al:13}. 

This motivated the definition of the larger class $\VPSPACE$, 
which mirrors $\PSPACE$ in the algebraic setting. 
This class was introduced by Poizat~\cite{poizat:08} and Koiran and Perifel~\cite{koir-perifel-R}. 
Let $(f_n)$ be of exponential format. We denote 
\begin{equation}\label{eq:f_n_c}
 f_n = \sum_k c_n(k) x_1^{k_1}\cdots x_{u(n)}^{k_{u(n)}}
\end{equation}
and assume that the degree of $f_n$ and the bitsize of the integer coefficients $c_n(k)$
are bounded by $2^{p(n)}$, where $p(n)$ is a polynomial.
Note that if $(f_n)\in\VPSPACE^0$, then the coefficients $(c_n(k_1,\ldots,k_{u(n)}))$ of~$f_n$,
as defined in~\eqref{eq:f_n_c}, 
define a family of integers of exponential bitsize. 

We define the {\em (total) coefficient function} of $(f_n)$ as the Boolean function, 
which takes an exponent vector $(k_1,\ldots,k_{u(n)})$ and
an index~$a$, all given in binary encoding (of polynomial size $O(u(n) p(n)$),  
and outputs the $a$-th bit in the binary expansion of $c_n(k)$,
with the convention that the sign of $c_n(k)$ is the $0$-th bit. 

The class $\VPSPACE^0$ is defined as the set of all families $(f_n)$ of 
exponential format such that its coefficient function lies in $\PSPACE/\poly$. 
The class $\VPSPACE^0$ is closed under many natural operations. 
Moreover, it contains the sequence of multivariate resultants, 
which follows from the fact~\cite{canny:88} that the multivariate resultants 
can be evaluated in $\PSPACE$; see also~\cite{grenet-et-al:13}. 

Replacing $\PSPACE$ by the counting hierarchy $\CH$, we arrive at the following central definition.
This definition implicitly appears in Koiran and Perifel~\cite{koir-per:11}. 

\begin{defn}\label{def:VCH}
The \emph{algebraic counting complexity class $\VCH^0$} 
consists of the families $(f_n)$ of multivariate polynomials over~$\Z$ of exponential format 
with the property that the coefficient function of $f_n$ is in $\CH/\poly$. 
The \emph{uniform algebraic counting complexity class $\VCH^0(\U)$} 
is obtained when additionally requiring that the coefficient function of $f_n$ is in $\CH$. 
\end{defn}

Using the terminology of \cref{se:int-def-CH}, 
we can rephrase~\cref{def:VCH} as follows: 
a family $(f_n)$ of multivariate integer polynomials of exponential format 
lies in the class $\VCH^0$ iff the corresponding family $c_n(k)$ of coefficients  
is definable in $\CH/\poly$.

It is clear that $\VCH^0 \subseteq \VPSPACE^0$ since $\CH\subseteq\PSPACE$.
In \cref{se:VCH-closure} we will present powerful criteria to show that 
families lie in $\VCH^0$.

Suppose we have a family $(f_n)$ of exponential format with $f_n\in \Z[x_1,\ldots,x_{u(n)},y_1,\ldots,y_{v(n)}]$. 
We view~$f_n$ as a polynomial in the $y$-variables with coefficients in $\Z[x_1,\ldots,x_{u(n)}]$ 
and thus write 
$$
 f_n(x,y) = \sum_\ell g_\ell(x_1,\ldots,x_{u(n)}) y_1^{\ell_1}\cdots y_{v(n)}^{\ell_{v(n)}} .
$$ 
Let us fix a map $n\mapsto \ell(n)$, which selects a monomial in $y$ in $f_n$. 
Then we call the family $(g_{\ell(n})$ the \emph{corresponding coefficient family}.

The following easy observation is extremely useful. 

\begin{cor}\label{cor:VCH-closed-coeff}
The class $\VCH^0$ is closed under taking coefficient families:
$(f_n)\in\VCH^0$ implies that $(g_{\ell(n)}) \in\VCH^0$ for any selection map $\ell$. 
Moreover, if the selection map $1^n\to\ell(n)$ is polynomial time computable and $(f_n)\in\VCH^0(\U)$, 
then $(g_{\ell(n)}) \in\VCH^0(\U)$.
\end{cor}

\begin{proof}
Suppose $f_n$ is given as in~\eqref{eq:f_n_c} by (using shorthand notation)
$$
 f_n(x,y) = \sum_{k,\ell} c_n(k,\ell) x^{k}\, y^{\ell} . 
$$
Then 
$g_\ell(x) = \sum_k c_n(k,\ell) y^{\ell}$
and
$g_{\ell(n)}(x) = \sum_k c_n(k,\ell(n)) y^{\ell}$
after fixing a map $n\mapsto \ell(n)$. 
If $(f_n)\in\VCH^0$, the language 
$L:=\{ (1^n, k,a) \mid\mbox{ the $a$-th bit of $c_n(k,\ell(n))$ equals $1$} \}$ 
is in $\CH/\poly$ since the coefficient function of $f_n$ is in $\CH/\poly$ 
and we can view the binary encoding of $\ell(n)$ as an advice of polynomial length.
Moreover, if $\ell$ is polynomial time computable, 
we do not need the advice and get $L\in\CH$ if $(f_n)\in\VCH^0(\U)$. 
\end{proof}

Investigating the ``closures'' of algebraic complexity classes 
recently has received a lot of attention;
see~\cite[\S 4.7]{buerg-survey:24} for references. 
Unfortunately, it is unknown whether the closure $\overline{\VP^0}$ of $\VP^0$
is strictly larger than $\VNP^0$, or whether the implication 
$ \VNP^0 \subseteq \overline{\VP^0} \Longrightarrow  \VNP^0=\VP^0$ holds.
A pleasant feature of $\VCH^0$ is that it coincides with its closure. 
Let us state this in precise form.

As before, we focus on families of multivariate integer polynomials of exponential format.
Let us define $\overline{\VCH^0}$ as the set of families~$(f_n)$ 
such that there exists a family $(F_n)\in\VCH^0$, depending on an extra indeterminate~$t$,  
with the property that for all~$n$ either $f_n=0$, or 
\begin{equation}\label{eq:TC}
 F_n(x,t) = f_n(x) t^{e_n} + t^{e_n+1} G_n(x,t) ,\quad f_n\ne 0 , 
\end{equation}
where $G_n(x,t)\in\Z[x,t]$. 
In other words, $f_n$ is the \emph{trailing coefficient} of $F_n$ 
and we denote $\tc_t(F_n) := f_n$.  
The exponent $e_n \deg F_n$ is bounded by $2^{p(n)}$ by~\cref{def:VCH}.
Note that the trailing coefficient can be seen as a limit as follows:
$$
 f_n(x) = \lim_{\tau\to 0} \tau^{-e_n} F_n(x,\tau) .
$$
Analogously, we define the class $\overline{\VCH^0(\U)}$.

\begin{cor}\label{cor:VCH-border}
We have $\overline{\VCH^0}=\VCH^0$ and $\overline{\VCH^0(\U)}=\VCH^0(\U)$
\end{cor}

\begin{proof}
This is an immediate consequence of~\cref{cor:VCH-closed-coeff}. 
\end{proof}

\subsection{Interpolation} 

The following applications of~\cref{th:closure-CH} already appear in~\cite{koir-per:11}. 
We present the proofs since they are important for understanding what is going on.

Let $\sigma_k(x_1,\ldots,x_N)$ denote the $k$-th elementary symmetric polynomial 
in the variables $x_i$. Then  the family given by 
$a_n(k) :=\sigma_{k}(1,2,3,\ldots,2^n)$ with $0\le k \le 2^n$ is definable in $\CH$;  
see~\cite[Cor~3.12]{buerg:09}.

Consider now the Vandermonde matrix 
$V_N(x_1,\ldots,x_N) := [x_j^{k-1}]_{1\le j,k\le N}$. 
It is well known that $\det(V_N) = \prod_{\a<\b} (x_\b-x_\a)$.
The inverse of $V_N$  is concisely described by the coefficients of the 
Lagrange polynomials for $j=1,\ldots,N$, 
\begin{equation*}
 L_j(t) = \sum_{k=1}^N L_{jk} x^{k-1} := \frac{f(t)}{(t-x_j) f'(x_j)} ,
\end{equation*}
where $f(t) := \prod_{i=1}^N (t- x_i)$ and 
$f'(x_j) = \prod_{i\ne j} (x_j-x_i)$. 
By construction, we have $L_j(x_i) = \delta_{ij}$, which means that 
$V_N \cdot [L_{jk}] =I_N$. Therefore, the matrix $[L_{jk}]$ is  the inverse of $V_N$.
We conclude that the coefficients of the adjoint $\adj(V_N) := \det(V_N) V_N^{-1}$ 
are given by the coefficients of the polynomial 
\begin{equation}\label{eq:lagrange}
 \det(V_N) L_j(t)  = \frac{\det(V_N)}{f'(x_j)} \cdot \frac{f(t)}{t-x_j} 
 = (-1)^{N-j}\prod_{\substack{\a<\b \\ \a,\b\ne j}} (x_\b-x_\a) \cdot \prod_{i \ne j} (t-x_i).
\end{equation}
We note that the coefficients of the right-hand side polynomial are given 
by elementary symmetric polynomials. 

Now we take $N=2^n$ and substitute $x_k$ by $k$ to obtain the family of integer matrices 
$M_n := V_{2^n}(1,2,3,\ldots,2^n)$. 
We note that the adjoint 
$\adj(M_n) := \det(M_n) M_n^{-1}$
is a matrix over~$\Z$. 

The following result was already derived from \cref{th:closure-CH} in~\cite{koir-per:11}.
A proof also appears in~\cite{andrews-et-al-HN:26} (in quite clumsy form), 
whose authors apparently were not aware of the references~\cite{buerg:09, koir-per:11}. 
  
\begin{cor}\label{cor:vandermonde}
\begin{enumerate}
\item The family of determinants $(\det (M_n))$ is definable in $\CH$.
\item The entries of the adjoint $\adj(M_n)$,
$k,\ell \mapsto \adj(M_n)_{k,\ell}$, where $1\le k,\ell \le 2^n$, 
form a family of integers definable in $\CH$.
\end{enumerate}
\end{cor}

\begin{proof}
1. The well known formula $\det(V_N) = \prod_{j<k} (x_k-x_j)$ implies $\det(M_n) = \prod_{j<k} (k-j)$. 
This, together with~\cref{th:closure-CH} implies the first assertion.

2. We use formula~\eqref{eq:lagrange} on Lagrange interpolation. 
The family of products $\prod_{\a,\b}(\b-\a)$, 
over all pairs $(\a,\b)$ such that $0\le \a,\b \le 2^n$ and $\a,\b\ne j$, is defined in $\CH$
by~\cref{th:closure-CH}. 
Moreover, the coefficients of the product 
$\prod_{i \ne j} (t-i)$ over $1\le j \le 2^n$ with $i\ne j$ are given by elementary symmetric 
polynomials evaluated at these integers~$j$. 
We already noted above that this provides families of integers definable in $\CH$
(\cite[Cor~3.12]{buerg:09}). 
Again applying~\cref{th:closure-CH} for the product of two integers completes the proof. 
\end{proof}

\subsection{Closure properties of \texorpdfstring{$\VCH^0$}{VCH0a+b=c}}\label{se:VCH-closure}

We first show that the class $\VCH^0$ is closed under taking exponential iterated sums and iterated products.

\begin{cor}\label{th:VCH-closed-IS-IP}
Let $(f_n)\in \VCH^0$ with $f_n\in \Z[x_1,\ldots,x_{u(n)+v(n)}]$ 
and $p(n)$ be a polynomial. Define 
\begin{eqnarray*}
 g_n(x_1,\ldots,x_{u(n)}) &:=& \sum_{y_1=0}^{2^{p(n)}-1}\ldots \sum_{y_{v(n)}=0}^{2^{p(n)}-1}  f_{n} (x_1,\ldots,x_{u(n)},y_{1},\ldots,y_{v(n)}) , \\ 
 p_n(x_1,\ldots,x_{u(n)}) &:=& \prod_{y_1=0}^{2^{p(n)}-1}\ldots \prod_{y_{v(n)}=0}^{2^{p(n)}-1}  f_{n} (x_1,\ldots,x_{u(n)},y_{1},\ldots,y_{v(n)}) ,
\end{eqnarray*}
Then the families $(g_n)$ and $(p_n)$ lie in $\VCH^0$.
Moreover, if  $(f_n)\in \VCH^0(\U)$, then $(g_n)$ and $(p_n)$ lie in~$\VCH^0(\U)$.
\end{cor}

\begin{proof}
This is an immediate consequence of \cref{th:closure-CH} and \cref{re:clos-SP-CH}.
\end{proof}

It will be convenient to have an equivalent characterization of $\VCH^0$
in terms of evaluation. 
The following definition is from~\cite{koir-per:11}.

\begin{defn}\label{def:eval-inVCH}
Let $(f_n)$ be a family of exponential format. 
We say that $(f_n)$ {\em can be evaluated at integers points in $\CH$} iff
the language 
$\{ (1^n, x_1,\ldots,x_{u(n)}, a) \mid\mbox{ the $a$-th bit of $f_n(x)$ equals $1$} \}$
is in $\CH$; here the $x_i$ denote integers encoded in binary. 
We make the analogous definition for $\CH/\poly$. 
\end{defn}

If we restrict to integers $0\le x_i < 2^{p(n)}$, where $p(n)$ is a polynomial, then 
the condition in \cref{def:eval-inVCH} states that the 
family $(f_n(x))$ of integers is definable in $\CH$. 

We next derive an equivalent characterization of $\VCH^0$. 
Its  proof essentially appears in~\cite[Thm.~3]{koir-per:11}, even though the characterization is 
not explicitly stated there. 

\begin{prop}\label{th:equ-char-VCH}
Let $(f_n)$ be a family of exponential format.
Then $(f_n)$ is in $\VCH^0$ iff $(f_n)$ can be evaluated at integer points in $\CH/\poly$.
Moreover, 
$(f_n)$ is in $\VCH^0(\U)$ iff $(f_n)$ can be evaluated at integers points in $\CH$.
\end{prop}

\begin{proof}
We first note that by repeated squaring, we have 
$\tau(x_1^{k_1}\cdots x_m^{k_m}) = O(mp(n) + \sum_i \log k_i)$.
Suppose now $(f_n) \in \VCH^0$ and let $f_n$ be as in~\eqref{eq:f_n_c}.
If we put 
$$
 a_n(x,k) := x_1^{k_1}\cdots x_{u(n)}^{k_{u(n)}} ,
$$
then we can write 
$$
 f_n(x) =\sum_k c_n(k) a_n(x,k) ,
$$
where the sum is over all $(k_1,\ldots,k_{u(n)})$ such that $0\le k_i < 2^{p(n)}$. 
The family $(a_n(x,k))$ is definable in $\CH$ by \cref{pro:slp-2-CH}.
Recall that we think of the integers $x_i$ and $k_i$ as being encoded in binary 
with a bitsize polynomially bounded in~$n$. 
Since  $(f_n) \in \VCH^0$, the family $(c_n(k))$ is definable in $\CH/\poly$. 
Here we used~\cref{th:closure-CH} and \cref{re:clos-SP-CH} twice: 
first to infer that the family $(c_n(k) a_n(x,k))$ of products of two integers is definable in $\CH/\poly$. 
Then to infer via the exponential sum
that the family $(f_n(x))$ is definable in $\CH/\poly$. 
This proves one direction of the proposition in the 
nonuniform setting. 
The argument also applies to the uniform situation.

For the the other direction, we use interpolation at integer points.
Recall the Vandermonde matrix 
$M_{p(n)}=V_{2^n}(1,2,3,\ldots,2^{p(n)})$.  
We use that $f_n(x) = \sum_k c_n(k) x^k$ for all 
$x\in\{1,2,3\ldots,2^{p(n)}\}^{u(n)}$ and express this 
by the matrix-vector multiplication
$$
 [f_n(x)] _x = \big(M_{p(n)}\big)^{\otimes u(n)} \cdot [c_n(k)]_k ,
$$
where $x$ and $k $ run over the $u(n)2^{p(n)}$ many elements of $\{1,2,3\ldots,2^{p(n)}\}^{u(n)}$.
This implies
$$
 \big(\det(M_{p(n)})\big)^{u(n)}\, [c_n(k)]_k 
 = \adj\Big(\big(M_{p(n)}\big)^{\otimes u(n)}\Big) \ [f_n(x)]_x .
$$
\cref{cor:vandermonde} tells us that the coefficients of the adjoint of $M_{p(n)}$ 
are definable in $\CH$. \cref{th:closure-CH} implies that this also applies 
to the adjoint of the $u(n)$-fold tensor power of $M_{p(n)}$.
By assumption, the $f_n(x)$ can be evaluated at integer points in $\CH/\poly$. 
With~\cref{th:closure-CH}, we conclude that 
$\big(\det(M_{p(n)})\big)^{u(n)}\, [c_n(k)]_k$ provides a family of integer points 
definable in $\CH/\poly$. 

From~\cref{cor:vandermonde},  
we also know that $\big(\det(M_{p(n)})\big)^{u(n)}$ is definable in $\CH$. 
Finally, we apply the second part of~\cref{th:closure-CH} on integer division 
to infer that $[c_n(k)]_k$ is definable in $\CH/\poly$.
This shows the reverse direction in the nonuniform situation.
The argument also applies to the uniform situation.
\end{proof}

\begin{cor}\label{cor:VCH-composition}
The class $\VCH^0$ is closed under composition in the following sense. 
Suppose we have a family $(f_n(x_1,\ldots,x_{u(n)})) \in \VCH^0$, in which we substitute $x_i$ 
by a polynomial $g_{n,i}$ with the property that for each $i=1,\ldots,u(n)$, 
the families $(g_{n,i})_n$ are in $\VP^0$. Then the resulting family given by 
$h_n:=(f_n(g_{n,1}(x),\ldots,g_{n,u(n)}(x))$ is in $\VCH^0$.
Moreover, if $f_n(x_1,\ldots,x_{u(n)}) \in \VCH^0(\U)$ and $(g_{n,i})_n \in \VP^0(\U)$ for all~$i$, 
then $(h_n)\in \VCH^0(\U)$.
\end{cor}

\begin{proof}
This is an immediate consequence of the characerization of $\VCH^0$ 
in terms of evaluation, given by~\cref{th:equ-char-VCH}.
\end{proof}

\begin{cor}
We have the chain of inclusions
$$
 \VNPnb^0 \subseteq \VCH^0 \subseteq \VPSPACE^0. 
$$
\end{cor}

\begin{proof}
It remains to prove the left inclusion. 
In~\cite[Thm.~4.1]{abpm:08} it was shown that the 
$\mathrm{BitSLP}$ problem lies in $\CH$. This immediately implies that 
every family $(f_n)\in\VPnb^0$ can be evaluated at integers points in $\CH/\poly$.
Using the characterization in~\cref {th:equ-char-VCH}, this shows that 
$\VPnb^0 \subseteq \VCH^0$.
By the closedness of $\VCH^0$ under exponential summation 
(\cref{th:VCH-closed-IS-IP} with $p(n)=1$) we derive that 
$\VNPnb^0\subseteq \VCH^0$. 
\end{proof}

We remark that \cref{th:closure-CH} also implies that the complexity class $\VPi^0$ 
of exponential products, introduced in~\cite{MR2298212}, 
is contained in $\VCH^0$, 

\subsection{Reducing degree and bitsize: repeated squaring}

We first explain how every family $(f_n)$ of polynomials of exponential format arises from 
a family $(F_n)$ of multilinear polynomials with coefficients in $\{-1,0,1\}$
with the property that the degree of $F_n$ is polynomially bounded in~$n$,  
and $f_n$ is obtained from $F_n$ by substitutions of the form 
$z\mapsto z^{2^{j}}$ or  $z\mapsto 2^{2^{j}}$. 
Using this, we will show that $\VP^0=\VNP^0$ implies $\VPnb^0=\VCH^0$; 
see~\cref{cor:collapse}. 

Suppose $(f_n)$ is a family of polynomials over $\Z$ of exponential format, 
say (see \eqref{eq:f_n_c})
\begin{equation*}
 f_n = \sum_k (-1)^{\e_n(k)}\, |c_n(k)| \, x_1^{k_1}\cdots x_{u(n)}^{k_{u(n)}} ,
\end{equation*}
where $\e_n(k)\in\{0,1\}$. 
The degree of $f_n$ and the bitsize of the coefficients are bounded by $2^{p(n)}$, 
where $p(n)$ is a polynomial. 
By introducing new variables (at most polynomially many in $n$), 
we are going to assign to~$f_n$ a multilinear polynomial $F_n$ 
with coefficients in $\{-1,0,1\}$
from which $f_n$ can be obtained back by substitutions as indicated above.

In order to define this procedure, we consider the binary expansion 
of the coefficients of $f_n$: 
$$
 |c_n(k)| = \sum_{a=0}^{2^{p(n)}-1}\, c_n(k,a) 2^a, \quad  c_n(k,a) \in \{0,1\} . 
$$
We also write $c_n(k) = (-1)^{\e_n(k)} |c_n(k)|$. 
We think of the integers
$0\le a, k_1,\ldots,k_{u(n)} < 2^{p(n)}$
as being encoded in binary by 
$(a_1,\ldots,a_{p(n)}) \in \{0,1\}^{p(n)}$ and  $(k_{i,1},\ldots,k_{i_n,p(n)}) \in \{0,1\}^{p(n)}$, 
respectively:
$$
 a= \sum_{\ell=1}^{p(n)} a_\ell \, 2^{\ell-1} ,\quad k_i= \sum_{j=1}^{p(n)} k_{i,j} \, 2^{j-1} .
$$ 
Now we introduce the new variables  $x_{i,j}$ and $y_\ell$ for $1\le i \le u(n), 1 \le j\le p(n)$ and $1\le \ell \le p(n)$.
Using this notation, we define the following \emph{multilinear} polynomial in these variables, 
\begin{equation*}
 F_n = \sum_{k} \sum_{a} (-1)^{\e_n(k)} c_n(k,a) \prod_{\ell=1}^{p(n)}  y_{\ell}^{a_\ell} \, \prod_{i=1}^{u(n)} \prod_{j=1}^{p(n)}   x_{i,j}^{k_{i,j}} ,
\end{equation*}
where the sum runs over all binary matrices 
$[k_{i,j}] \in \{0,1\}^{u(n) \times p(n)}$ and 
binary vectors 
$[a_{\ell}] \in \{0,1\}^{p(n)}$.
Note that the coefficients of $F_n$ are in $\{-1,0,1\}$. 
The point of this construction is that the substitutions
$$
 x_{i,j} \mapsto x_i^{2^{j-1}}, \quad y_\ell \mapsto 2^{2^{\ell-1}} 
$$ 
map $F_n$ back to $f_n$. Indeed,
$$
 \prod_{\ell=1}^{p(n)}  \big(2^{2^{\ell-1}}\big)^{a_\ell} = \prod_{\ell=1}^{p(n)}  2^{a_\ell 2^{\ell-1}} = 2^a , \quad 
 \prod_{j=1}^{p(n)}  \big(x_i^{2^{j-1}}\big)^{k_{i,j}}  = \prod_{j=1}^{p(n)}  2^{k_{i,j} 2^{j-1}} = 2^{k_i} .
$$

Let us summarize some of the properties of this construction.

\begin{lem}\label{le:f2F-VCH}
The following holds, with analogous statements in the uniform setting: 
\begin{enumerate}
\item $(F_n) \in \VP^0 \Longrightarrow (f_n) \in \VPnb^0$.
\item $(F_n) \in \VNP^0 \Longrightarrow (f_n) \in \VNPnb^0$.
\item The polynomials $f_n$ and $F_n$ have the same coefficient functions 
(up to a trivial rearrangement caused by our convention on the sign bit). 
\item  $(f_n) \in \VCH^0 \Longleftrightarrow  (F_n) \in \VCH^0$.
\end{enumerate}
\end{lem}

\begin{proof}
The arguments given below immediately extend to the uniform setting. 

1. If $(F_n) \in \VP^0$, then $\tau(F_n)$ is polynomially bounded in~$n$. 
Then $\tau(f_n)$ is as well polynomially bounded in~$n$, since 
$z^{2^{j}}$ can be computed from $z$ with $j$ multiplications by repeated squaring, 
Therefore, $(f_n) \in \VPnb^0$. 

2. This follows from part one and the definition of $\VNPnb^0$; 
see~\cite[Def.~4.4]{buerg-survey:24}. 

3. The coefficient function of $f_n$ 
takes as input $k,a$ in binary and outputs the $a$-th bit of $c_n(k)$, 
which is $c_n(k,a)$ (or $\e_n(k,a)$ if $a$ encodes the $0$-th bit).
The same applies to the the coefficient function of $F_n$. 

4. This follows from part three.
\end{proof}

The construction made becomes especially helpful in the case of 
a collapse of complexity classes. 

\begin{cor}\label{cor:collapse}\phantom{x}
\begin{enumerate}
\item If $\VP^0=\VNP^0$, 
 then $\CH\subseteq \Po/\poly$ and $\VPnb^0=\VCH^0$. 
\item If $\VP^0(\U)=\VNP^0(\U)$, 
  then $\CH = \Po$ and $\VPnb^0(\U)=\VCH^0(\U)$.
\end{enumerate}
\end{cor}

\begin{proof}
Suppose $\VP^0 = \VNP^0$. 
Then the family of permanents is in $\VP^0$, which implies  
$\PP\subseteq \Po/\poly$; see \cite[Lemma~2.13]{buerg:09}.
In turn, this implies the collapse of the counting hierarchy: 
$\CH\subseteq \FPo/\poly$; see \cite[Lemma~2.6]{buerg:09}.
Similarly, $\VP^0(\U)=\VNP^0(\U)$ 
implies that $\PP =\Po$ and hence $\CH = \Po$. 

Let now $(f_n) \in \VCH^0$. Then $(F_n) \in \VCH^0$ by \cref{le:f2F-VCH}(4).
Since we assume $\CH\subseteq\Po/\poly$, the coefficient function of~$F_n$
is in $\Po/\poly$. Valiant's criterion~\cite[Prop.~2.20]{buer:00-3} implies 
that $(F_n) \in \VNP^0$. (The proof there shows that the implication also 
holds for the constant-free class.) 
Since we assume $\VP^0 = \VNP^0$, we even get $(F_n) \in \VP^0$. 
Finally, \cref{le:f2F-VCH}(1) implies that $(f_n) \in \VPnb^0$. 
These arguments extend to the uniform setting.
\end{proof}

\subsection{Families of resultants in \texorpdfstring{$\VCH^0$}{VCH0}}

Recently, Andrews, Garg, and Schost~\cite{andrews-et-al-HN:26} proved 
that the multivariate resultant~$\Res$ 
can be evaluated at integer points in the counting hierarchy. 
We will explain the precise meaning of this statement in~\cref{re:expl-A} below. 
This will imply that certain families of resultants are in the class $\VCH^0$. 

The resultant $\Res_{d_0,\ldots,d_n}$ is a polynomial function $\C^p \to \C$ 
mapping the joint coefficient vector $u\in\C^p$ 
of the homogeneous polynomials $f_1,\ldots,f_s\in\C[x_1,\ldots,x_n]$ 
to the resultant $\Res_{d_0,\ldots,d_n}(f_0,\ldots,f_n)\in\C$, 
where $d_i =\deg d_i $ and $p=\sum_{i=1}^s \binom{n+d_i}{n}$.
In general, this does not define a family of exponential format since 
the number~$p$ of variables grows exponentially in $d_0,\ldots,d_n$. 
In order to fix this, we  bound the degrees~$d_i$. 
For our purposes it will be sufficient to assume $d_0=\ldots=d_n=2$, which implies $p=O(n^3)$. 
We can do so without loss of generality since
it is well-known that the restriction of $\HN_\C$ 
to systems of polynomials of degree at most two 
is still $\NP^0_\C$-complete; see~\cite{BCSS:98}.

We denote by $\cR_n :=\Res_{2,\ldots,2}$ the resultant of $n+1$ homogenous quadratic forms.  
The family $(\cR_n)$ is of exponential format, since the number of input variables 
of $\cR_n$ is 
$v(n) := n \binom{n+2 }{2} = O(n^3)$,  
$\deg\cR_n = n2^n$ and the bitsize of $\cR_n$ is bounded by $O(n2^n)$; 
see \cref{se:resultant}.

For systems of quadratic forms, the recent result~\cite{andrews-et-al-HN:26} 
exactly expresses the following.

\begin{thm}[Andrews, Garg, Schost]\label{th:res-eval-in-CH}
The family $(\cR_n)$ of resultants of quadratic systems can be evaluated at integer points in $\CH$. 
\end{thm}
 
From this we infer the following. 

\begin{cor}\label{cor:Res-in-VCH}
The  family $(\cR_n)$ of resultants of quadratic systems lies in $\VCH^0(\U)$. 
\end{cor}

\begin{proof}
Combine \cref{th:res-eval-in-CH} with \cref{th:equ-char-VCH}. 
\end{proof}

\begin{rem}\label{re:expl-A}
The result in~\cite{andrews-et-al-HN:26} is not restricted to bounded degrees~$d_i$. 
Let us consider polynomials~$f_i$ with integer coefficients of bitsize bounded by~$\ell\in\N$. 
The resultant defines a Boolean function $\cF_{\bn}\colon\{0,1\}^{m'_\bn} \to \{0,1\}^{m_\bn}$,  
which sends $(f_0,\ldots,f_n)$, given in binary encoding, 
to $\Res_{d_0,\ldots,d_n}(f_0,\ldots,f_n)$, in binary encoding. 
Here $\bn$ denotes the multi-index $(n,d_0,\ldots,d_n,\ell)$. 
Both the input bitsize $m'_\bn$ and output bitsize $m_\bn$ 
may be exponential in $\bn$. For  this reason, 
as in~\cref{def:eval-inVCH},
we define the following language (cp.~\cite[Def.~6.2]{andrews-et-al-HN:26}) 
$$
 L_{\cF_{\bn}} := \{ (1^{\bn}, \xi, a) \mid\mbox{ the $a$-th bit of $\cF_{\bn}(\xi)$ equals $1$} \} .
$$
\cite[Thm.~6.7]{andrews-et-al-HN:26} states that $L_{\cF_{\bn}} \in\CH$. 
\end{rem}

\begin{rem}
In order to avoid the restriction to bounded degrees~$d_i$, we may  
think of the homogenous polynomial $f_i$ as being given by an arithmetic circuit $\Gamma_i$ 
for its evaluation, which allows for a succinct encoding of $f_i$: 
the input size of $f_i$ will then be the number of complex input parameters of $\Gamma_i$. 
Using~\cref{th:equ-char-VCH} and \cref{cor:VCH-composition}, 
one can show that the resulting families of multivariate resultants are in $\VCH^0$. 
\end{rem}

\section{Reducing Hilbert's Nullstellensatz to multivariate resultants}\label{se:redu}

In the paper~\cite{andrews-et-al-HN:26} an ingenious reduction from Hilbert's Nullstellensatz Problem
to the multivariate resultant, following Ierardi~\cite{ierardi:89}, is used.
This relies on Canny's~\cite{canny-gcp:89} concept of the generalized characteristic polynomial, 
see also \cite{canny-conf:88} and \cite{renegar-I:92}.
The reduction in~\cite{andrews-et-al-HN:26} applies to polynomials over fields, which 
are finite extensions of the prime fields, 
and the reduction is carried out in the model of Turing machines. 

Our goal is to describe a reduction in the BSS-model over $\C$ (or in the model of algebraic circuits), 
which applies to any polynomials with complex coefficients. 
This causes additional difficulties, since it is unknown whether the projective version of the 
Hilbert Nullstellensatz problem is $\NP_\C$-complete.\footnote{Given a system of homogeneous polynomials, is there 
a solution in complex projective space?}
By contrast, the corresponding problem in the Turing model is $\NP$-hard for integer polynomials, 
even for systems of quadratic forms~\cite{grenet-et-al:13}, . 

We describe our reduction in terms of families of polynomial size oracle circuits, 
which have the ability to evaluate the resultants $\cR_n$, and which can extract trailing coefficients. 
The power of such reductions is captured by the complexity class~$\VCH^0$.
For this, it is essential that $(\cR_n)\in\VCH^0$ (\cref{cor:Res-in-VCH}).

The point is that under the hypothesis $\VP^0=\VNP^0$, 
the class $\VCH^0$ collapses with $\VPnb^0$ by~\cref{cor:collapse} 
and we are left with a polynomial time reduction.
The notions of generic quantifiers and witness sequences allow 
to avoid the use of randomness; see~\cref{se:witness-s}. 

Let us describe these oracle circuits in more detail. 
There is no need to be completely formal, since we will only 
rely on the very particular procedure in the proof of~\cref{pro:NPC-in-FP-VCH}. 

The input to $\cC_\bn$ is a vector $x\in\C^{u(\bn)}$, where $u(\bn)$ is 
polynomially bounded in the multi-index~$\bn$. 
The circuit~$\cC_\bn$ is a constant-free arithmetic circuit
of size polynomially bounded in~$\bn$, which, in addition to the usual arithmetic gates,  
is endowed with two additional powerful features.

(1) $(\cC_\bn)$ can rely on the family $(\cR_n)$ of multivariate resultants of systems of quadratic forms. 
More specifically, $\cC_\bn$ has \emph{oracle gates}, labelled by $n$ and with fan-in $v(n)$.
This gate outputs the composition $\cR_n(h_1,\ldots,h_{v(n)})$, 
where $h_1,\ldots,h_{v(n)}$ are the polynomials computed at the gates fed into the gate. 
(One has to ensure that the $h_i$ are homogeneous of degree~$2$.)

(2) $\cC_\bn$ has fan-in one \emph{gates for computing trailing coefficients}: 
they extract the trailing coefficient $\tc_t(f)$ of the polynomial $f$ fed into the gate, 
with respect to a distinguished variable~$t$, whose name is a label of the gate; 
see~\eqref{eq:TC}. 

We say that the family $(\cC_\bn)$ is \emph{uniform}, 
if the description of $(\cC_\bn)$ can be computed by a Turing machine in polynomial time in~$\bn$. 

\subsection{Reducing \texorpdfstring{$\HN_\C$ }{HNC} to evaluating trailing coefficients of resultants}\label{se:RED-MAIN}
We denote by $u\in\C^p$ the joint coefficient vector of $f_1,\ldots,f_s\in\C[x_1,\ldots,x_n]$, 
where $d_i =\deg d_i \ge 1$;  
note $p=\sum_{i=1}^s \binom{n+d_i}{n}$. 
In the proof of~\cref{pro:NPC-in-FP-VCH}, we will have $d_i=2$, but it is clearer to work here in more generality. 
We note that 
\begin{equation}\label{eq:tub}
 s+ n \le s(n+1) \le p, \quad d:= \max_i d_i < p .
\end{equation}

The features of our reformulation of the reduction in~\cite[Prop.~4.4]{andrews-et-al-HN:26} 
are summarized in the next result. 
Let us point out that working in the algebraic model makes things more transparent and removes 
a lot of the clutter resulting from dealing with the details of bit computations.

\begin{prop}\label{pro:NPC-in-FP-VCH}
There is a uniform family of arithmetic circuits 
of size polynomially bounded in~$n$, 
endowed with oracle gates for evaluating multivariate resultants 
and gates for computing trailing coefficients,  
which on input the coefficient vector $u\in\C^{p}$ of 
$f_1,\ldots,f_s\in\C[x_0,\ldots,x_n]$ 
of degree $d_i\le 2$, 
computes for all $i\in [2p+1]$ and $j\in [n]$ a polynomial $\rho_{ij}\in\C[v]$ 
with the property that the zero set $V(f_1,\ldots,f_s)$ is nonempty iff
$\exists j\in [n]\ \rho_{ij}(0) = 0$ holds for the majority of the $i\in [2p+1]$, that is,
$$
 \#\{ i \in [2p+1] \mid \exists j \in [n] \rho_{ij}(0) = 0 \} > p .
$$
\end{prop}

From this result, 
combined with $(\cR_n)\in\VCH^0(\U)$ (\cref{cor:Res-in-VCH}) 
and the properties of the complexity class $\VCH^0$ developed in \cref{se:VCH}, 
we will be able to quickly prove~\cref{th:MAIN-u}. 

The reduction consists of several steps, all of which are known techniques.
We will state the mathematical insights behind each step in the form of lemmas. 

The input consists of $f_1,\ldots,f_s\in\C[x_1,\ldots,x_n]$ of degree at most~$d$, 
given by their coefficient vector $u\in\C^p$. 
We denote by $X_u :=V(f_1,\ldots,f_s)\subseteq\C^n$ the zero set of $f_1,\ldots,f_s$. 
Our goal is to detect whether $X_u$ is nonempty. 

Let $a_{ik}$ and $b_{i\ell}$ be new variables, which take the role of ``generic parameters''.
For $1\le j\le n$ we compute  the linear combinations of $f_1,\ldots,f_s$, 
$$
  g_{ij} := \sum_{k=1}^s a_{ik} f_k , \quad\mbox{$i=1,\ldots,j$} , 
$$
and consider the affine linear polynomials 
$$
  g_{ij} :=  b_{i0} + \sum_{\ell=1}^n b_{i\ell} x_\ell , \quad \mbox{$i=j+1,\ldots,n$} .
$$

\begin{lem}\label{le:gen-lin-comb}
Suppose that $X_u$ is nonempty of dimension~$m$. 
For Zariski-almost all specializations of the parameter systems 
$(a,b)$ to values in $ \C^{j\times s} \times \C^{(n-j)\times s}$, 
the following holds:
\begin{enumerate}
\item If $j\le n-m$, then all irreducible components of $V(g_{1j},\ldots,g_{jj})$ have dimension $n-j$.
Moreover, $V(g_{1j},\ldots,g_{jj}, g_{j+1j},\ldots,g_{nj})$ is finite and nonempty. 
\item If $j=n-m$, then every irreducible component of $X_u$ of dimension~$m$ is an irreducible component of $V(g_{1j},\ldots,g_{jj})$.
Moreover, $X_u\cap V(g_{1j},\ldots,g_{nj})$ is nonempty. 
\end{enumerate}
Here, in order to ease notation, we still denote the resulting polynomials by $g_{ij}$ after specializing. 
\end{lem}

\begin{proof}
This is well known. 
For (1), e.g.; see the proof of~\cite[Lemma~4.14]{buer:00-3}.
The first statement of~(2) is an immediate consequence of (1).
For the second statement of~(2),  we use the following fact: 
If $W$ is an irreducible component of $V(g_{1j},\ldots,g_{jj})$ of dimension $n-j$, then 
Zariski-almost all affine linear subspaces with dimension~$j$ intersect $W$, e.g., 
see~\cite[Chap.~9]{cox-little-oshea:92}
\end{proof}

Since we want to use the resultant, we homogenize the $g_{ij}$ with respect to a fixed 
new variable~$x_0$. Note that $\deg g_{ij} \le d$. We define 
$G_{ij} := x_0^d g_{ij}(x_1/x_0,\ldots,x_n/x_0)$, which is homogeneous of degree~$d$.
We also define the homogenization 
$F_k := x_0^d f_k(x_1/x_0,\ldots,x_n/x_0)$ of~$f_k$, 
for $k=1,\ldots,s$, which has degree~$d$. 
By~\cref{le:gen-lin-comb}, 
$V(g_{1j},\ldots,g_{nj})$ is finite and contains a point of $X_u$ if $j=n-m$. 
However, as illustrated in \cref{se:deform}, 
$V(G_{1j},\ldots,G_{nj})$ may not be finite.
For coping with this, we apply now the deformation methods explained in~\cref{se:deform}. 

In the following, we think of  $1\le j\le n$ as being fixed. 
We introduce a new variable~$t$ (taking the role of a generic deformation parameter)
and define for $i=1,\ldots,n$ the ``perturbed'' polynomials 
$$
 \hG_{ij} := G_{ij} + t (x_i^d -x_0^d) \in \C[u,t,a,b][x_0,\ldots,x_n] ,
$$
which are homogenous of degree~$d$ in the $x$-variables.
We also introduce new variables $v,w,c$ 
and define the polynomial 
\begin{equation}\label{eq:def-H}
  H_j := v x_0^d + \sum_{i=1}^n w^i x_i^{d}  + \sum_{k=1}^s c^k F_k \in\C[u,v,w,c][x_0,\ldots,x_n] , 
\end{equation}
which is homogenous of degree~$d$ in the $x$-variables.
Consider now the resultant 
\begin{equation}\label{eq:def-R}
 R_j := \Res_{d,d,\ldots,d}(H_j,\hG_{1j},\ldots, \hG_{nj}) \in \C[u,v,w,t,a,b,c] 
\end{equation}
with respect to $x_0,\ldots,x_n$.
We note that the number $q$ of components of the vector $(a,b,c)$ is bounded as 
\begin{equation}\label{eq:k-bound}
 q =js + (n-j)n +1 \le ns + n^2 -n + 1\le n(s+n) \le pn \le p^2 . 
\end{equation}

\begin{lem}\label{le:deg-res}
The total degree of $R_j$ in $u,v,w,t,a,b,c$ is less than $(2p)^{n+1}$. 
The bit size of~$R_j$ is at most~$p^{O(n)}$.
\end{lem}
 
\begin{proof}
By construction, $\hG_{ij}$ has degree at most two in $u,t,a,b$.
Moreover, $H_j$ has degree at most $\max\{n,s+1\} \le s+n$ in $u,v,w,c$.
Recall from~\cref{se:resultant} that the 
multivariate resultant $R_{d_0,\ldots,d_n}$ of forms $F_0,\ldots,F_n$ 
with $d_i =\deg(F_i)$ is homogeneous of degree $d_0\cdots d_n/d_i$ in the variables of~$F_i$. 
From this we conclude that the total degree of $R_j$ is at most 
$$
 2d^n + n(s+n)(2d)^{n-1} \le (2d)^{n-1} \big(2d + n(s+n) \big) \le  (2d)^{n-1} 2p^2 \le (2d)^{n+1}, 
$$
where we used~\eqref{eq:k-bound} for the rough bound 
$2d +n(s+n) \le 2p + np \le 2p^2$. 
This proves the degree bound. 

As for the bit size, we observe that the coefficients of $\hG_{ij}$ and  $H_j$ 
are in $\{0,\pm 1\}$. 
We omit the argument for the upper bound of the bitsize of $R_j$. 
The stated bound certainly is very generous. 
(Compare with the upper bound $O(n d^n \log d)$ for 
$R_{d_0,\ldots,d_n}$ following from~\cite{sombra:04}.)
\end{proof}

We now consider the trailing coefficient 
$$
\tc_t(R_j) \in \C[u,v,w,a,b,c]
$$ 
of $R_j$ with respect to~$t$. 
From this we take its trailing coefficient with respect to the variable~$w$:
$$
 \rho_j := \tc_w(\tc_t(R_j)) \in \C[u,v,a,b,c] . 
$$
It will be relevant whether the evaluation $\rho_m(u,0,a,b,c)$ of the polynomial $\rho_j$ 
at $v=0$ equals zero. 

The following result summarizes what has been achieved by this construction. 

\begin{lem}\label{le:def-u-resultant}
\begin{enumerate}
\item Suppose that $X_u$ is nonempty of dimension~$m$. Then: 
$$
 \forall^\ast (a,b,c)\quad \rho_m(u,0,a,b,c) = 0 .
$$

\item Suppose that $X_u$ is empty. Then: 
$$
 \forall j \in [n]\ \forall^\ast (a,b,c)\quad \rho_j(u,0,a,b,c) \ne 0 .
$$
\end{enumerate}
\end{lem}

\begin{proof}
1. Suppose $m=\dim(X_u)$ and put $j=n-m$.
By~\cref{le:gen-lin-comb}(2) we know that for almost all specializations of $(a,b)$ that 
$V(g_{1j},\ldots,g_{nj})$ is finite and contains a point $\xi\in X_u=V(f_1,\ldots,f_s)$. 
In particular, $\xi$ is  an isolated zero of $V(g_{1j},\ldots,g_{nj})$. 
\cref{pro:curves+limits}(3) tells us that $\xi$ occurs among 
the affine limit points of 
$Z:= V(\hG_1,\ldots, \hG_n) \subseteq \proj^n\times \C$.
We have $\xi_0\ne 0$ and 
$F_1(\xi)=0,\ldots, F_s(\xi)=0$. 
With~\eqref{eq:def-H} we therefore get
$$
   H_m(\xi) := v \xi_0^d + \sum_{i=1}^n w^i \xi_i^{d}  \ne 0 . 
$$
Its trailing part with respect to $w$ equals
\begin{equation}\label{eq:trw}
  \tc_w(H_m(\xi)) = v \xi_0^d \ne 0 ,
\end{equation}
and this polynomial vanishes at $v=0$. 
\cref{pro:TC-RES} tells us that ($c'\in\C^*$) 
\begin{equation}\label{eq:tcp}
 \tc_t(R_m) = c' \prod_{j=1}^r H_m(q^{(j)})  \ne 0,
\end{equation}
where  $q^{(1)},\ldots, q^{(r)}$ are the affine limit points of 
$Z\subseteq \proj^n\times \C$
as defined in \cref{pro:curves+limits}. 
By passing to the trailing coefficients with respect to $w$, 
and using multiplicativity, we get 
\begin{equation}\label{eq:tcw-prod}
 \tc_w(\tc_t(R_m)) = c' \prod_{j=1}^r \tc_w\big( H_m(q^{(j)}) \big) \ne 0.
\end{equation}
This vanishes when evaluated at $v=0$:
the reason is~\eqref{eq:trw} and the fact that 
$\xi$ occurs among the $q^{(j)}$, 
as already noted before. 
This proves the first assertion.

2. Now assume that $X_u$ is empty. Let $\xi$ be any one of the affine limit points 
$q^{(1)},\ldots, q^{(r)}$ that appear in~\eqref{eq:tcp}. Then $\xi_0\ne 0$ and 
$F_\kappa(\xi)\ne 0$ for some $\kappa$ since $X_u=V(f_1,\ldots,f_s)=\varnothing$. 
Therefore, 
$$
   H_m(\xi) := v \xi_0^d + \sum_{k=1}^s \gamma^k F_k(\xi) + \sum_{i=1}^n w^i \xi_i^{d}  ,
$$
where $\gamma$ denotes a specialization of $c$. 
Since 
$$
 v \xi_0^d + \sum_{k=1}^s \gamma^k F_k(\xi) \ne 0 ,
$$
we obtain 
$$
 \tc_w(\tc_t(H_m(\xi)))  =\tc_w(H_m(\xi)) = v \xi_0^d + \sum_{k=1}^s \gamma^k F_k(\xi) \ne 0.
$$
When evaluated at $v=0$, this gives 
$\sum_{k=1}^s \gamma^k F_k(\xi)$,
which is nonzero for all but finitely many~$\gamma$. 
(Note that $\xi$ only depends on $(\a,\b)$ and not on $\gamma$.
Moreover $F_\kappa(\xi) \ne 0$ for some $\kappa$.) 
We conclude as for~\eqref{eq:tcw-prod}
$$
 \rho_j = \tc_w(\tc_t(R_j)) = c' \prod_{j=1}^r \tc_w\big( H_m(q^{(j)}) \big) \ne 0,
$$ 
and this does not vanish for almost all specializations of $(a,b,c)$, 
which proves the second assertion.
\end{proof}

The previous lemma immediately implies the following resultant based criterion for testing 
whether the affine zero set of given polynomials $f_1,\ldots,f_s$ is nonempty. 
(Note that the quantifiers $\exists j\in [m]$ and $\forall^\ast (a,b,c)$ can be interchanged.)

\begin{cor}\label{cor:feas-test}
For polynomials $f_1,\ldots,f_s\in\C[x_1,\ldots,x_n]$, given by their coefficient vector $u\in\C^u$
with zero set $X_u:=V(f_1,\ldots,f_s)$, we have 
$$
 X_u \ne \varnothing \Longleftrightarrow 
 \forall^\ast (a,b,c)\  \exists j\in [n]\ \rho_j(u,0,a,b,c) = 0 .
$$
\end{cor}

We next provide a first order formula for the quantifier-free part of the condition appearing 
on the right hand side of the criterion in~\cref{cor:feas-test}.

\begin{lem}\label{le:FOF}
Set $D:=(2p)^{n+1}$ and $M:=nD^3$. 
One can express the statement 
$$
 \exists j\in [n]\ \rho_j(u,0,a,b,c) = 0
$$ 
by a first order formula $F(u,a,b,c)\in\FC$ 
as in~\cref{th:cons-pws} 
with variables $u\in\C^p$ and $(a,b,c)\in \C^q$,
where $q\le n(s+n)\le p^2$, 
with atomic predicates 
given by $M$ many polynomials 
$R_{j\la\mu\nu}\in\Z[u, a,b,c]$, where $j\in [n]$, $0\le \la,\mu,\nu < D$, 
which have degree at most~$D$ and have 
integer coefficients of bit size at most~$\ell =D^{O(1)}$.
\end{lem}

\begin{proof}
The stated upper bound on $q$ is~\eqref{eq:k-bound}. 
From~\cref{le:deg-res} we see that 
the resultant $R_j \in \C[u,v,w,t,a,b,c]$ defined in~\eqref{eq:def-R}
has degree less than~$D$ and bit size at most $D^{O(1)}$. 
We expand the resultant as 
$$
 R_j = \sum_{0\le \la,\mu,\nu < D} R_{j\la\mu\nu}(u,a,b,c) \, v^{\la}  w^{\mu} t^{\nu} ,\quad 
 R_{j\la\mu\nu}(u) \in \Z[u,a,b,c] .
$$
We can express the statement  
$\rho_j(u,0,a,b,c) = \tc_w(\tc_t(R_j)) = 0$
as the conjunction of the following Boolean formulas
$$
\exists \la_0,\mu_0,\nu_0\  \forall \la,\mu\ \forall \nu < \nu_0\ 
  \big( R_{j\la_0\mu_0\nu_0} \ne 0 \wedge R_{j\la\mu\nu} = 0 \big)
$$
$$
\exists \la_1,\mu_1\  \forall \la'\ \forall \mu' < \mu_1\ 
  \big( R_{j\la_1\mu_1\nu_0} \ne 0 \wedge R_{j\la'\mu'\nu_0} = 0 \big)
$$
$$
   R_{j 0 \mu_1\nu_0} = 0 .
$$
Indeed, note that the  first formula expresses that 
$$
 \tc_t(R_j) = \sum_{\la,\mu} R_{j\la\mu\nu_0} v^{\la} w^\mu ,
$$
while the second formula expresses that 
$$
 \tc_w(\tc_t(R_j)) = \sum_{\la} R_{j\la\mu_1\nu_0} v^{\la} =  R_{j 0 \mu_1\nu_0}+ \sum_{\la>0} R_{j\la\mu_1\nu_0} v^{\la},
$$
which vanishes at $v=0$ iff $R_{j 0 \mu_1\nu_0}=0$. 

Therefore, we can express the statement  
$\exists j\in [n]\ \rho_j(u,0,a,b,c) = 0$
by a Boolean formula 
in terms of the $M$ Boolean values
$\e_{j\la\mu\nu} := \sgn\big(R_{j\la\mu\nu}(u)\big)$, 
where 
$\sgn(c):= 0$ if $c=0$ and $\sgn(c):= 1$ if $c\ne 0$; 
see~\eqref{eq:wf-formula}. 
\end{proof}

\begin {proof}[Proof of~\cref{pro:NPC-in-FP-VCH}]
We assume that $d=2$.
By combining~\cref{le:FOF} with~\cref{th:cons-pws}, 
we see that a witness sequence 
for the formula $F(u,a,b,c)$ 
can be computed by a division-free and constant-free 
arithmetic circuit $\Gamma$ of size 
$$
 (pq)^{O(1)} \log (MD) + O(\log\ell) \ \le\  p^{O(1)} .
$$
Moreover, by \cref{th:cons-pws},
$\Gamma$ can be computed in time 
polynomial in $p$ by a Turing machine.  

The procedure below describes a polynomial size uniform family 
of arithmetic circuits with oracle gates for $(\cR_n)$ 
and gates for computing trailing coefficients, 
which decides $\HN_\C$ for systems of polynomials of degree at most two.
\algorithm
{\bf input} $f_1,\ldots,f_s$ with coefficient vector~$u$\\
compute an arithmetic circuit $\Gamma$ as above\\
evaluate $\Gamma$ to compute a witness sequence
$((\a_\theta,\b_\theta,\c_\theta))_{\theta=1,\ldots,2p+1}$ of $F(u,a)$\\
{\bf for $\theta=1,\ldots,2p+1$} \\
\>     {\bf for $j=1,\ldots,n$} \\
\>\>         substitute $(a_\theta,b_\theta,c_\theta)$ by $(\a_\theta,\b_\theta,\c_\theta)\in\C^{j\times s}\times\C^{(n-j)\times s}\times \C$ \\
\> \>        compute $\hG_{1j},\ldots,\hG_{nj}, H_j$ \\
\>\>         {\bf by oracle call to $\cR_n$ get resultant} \\
\>\>         $R_j = \cR_n(H_j,\hG_{1j},\ldots, \hG_{nj}) \in \C[t,v,w]$\\
\>\>         {\bf by oracle call get trailing coefficient} $\tc_t(R_j) \in \C[v,w]$ \\
\>\>         {\bf by oracle call get trailing coefficient} $\rho_{\theta j} := \tc_w(\tc_t(R_j)) \in \C[v]$ \\
\> \>        {\bf evaluate} $\rho_{\theta j}(0)$ \\
\>        {\bf set} $\b_\theta :=1$ if $\exists j\in [n]\ \rho_{\theta j}(0)=0$ and $\b_\theta :=0$ otherwise\\
{\bf accept iff} $\b_1+\ldots+\b_{2p+1} > p$ 
\falgorithm
We claim that this procedure accepts iff $X_u=V(f_1,\ldots,f_s)$ is nonempty.
Indeed, we have 
$$
 X_u\ne \varnothing \Longleftarrow \forall^\ast (a,b,c)\  F(u,a,b,c) 
 \Longleftarrow \#\{ \theta \in [2p+1] \mid F(u,a_\theta ,b_\theta ,c_\theta ) \} > p,
$$
where the left equivalence is due to~\cref{cor:feas-test} and 
the right equivalence holds by the definition of witness sequences (\cref {def:pws}). 
This completes the proof of~\cref{pro:NPC-in-FP-VCH}. 
\end{proof}

\subsection{Proof of \texorpdfstring{\cref{th:MAIN-u}}{MainThm}}

We show the contraposition and assume that $\VP^0=\VNP^0$.
\cref{cor:collapse} implies that $\CH\subseteq\Po/\poly$ and $\VPnb^0 = \VCH^0$. 
\cref{cor:Res-in-VCH} states that $(\cR_n)\in\VCH^0(\U)$. 
We conclude that $(\cR_n)\in\VPnb^0$.  

The problem $\HN_\C$, restricted to 
input polynomials of degree at most two,  
is $\NP^0_\C(\NU)$-complete with respect to 
nonuniform polynomial time reductions.
Therefore, in order to conclude $\Po^0_\C(\NU)=\NP^0_\C(\NU)$, 
it remains to show that the oracle circuit $\cC_n$ 
described in the proof of~\cref{pro:NPC-in-FP-VCH}
can be implemented by polynomial size arithmetic circuits. 

This is clear for the oracle gate evaluating $\cR_n$ 
since we already know that $(\cR_n)\in\VPnb^0$.   
Consider now a gate computing the trailing coefficient $\tc_t(R_j)$. 
By \cref{cor:VCH-border}, the family of trailing coefficients are in $\overline{\VCH^0}=\VCH^0$.
Since we know that $\VPnb^0 = \VCH^0$,  their 
computations can also be implemented by an arithmetic circuit of polynomial size. 

In the uniform models, we argue analogously. 
When assuming $\VP^0(\U)=\VNP^0(\U)$, we get from~\cref{cor:collapse} 
the collapses $\CH=\Po$ and $\VPnb^0(\U)= \VCH^0(\U)$. 
\cref{cor:Res-in-VCH} states that $(\cR_n)\in\VCH^0(\U)$. 
Hence we conclude that $(\cR_n)\in\VPnb^0(\U)$.  
It is clear that the uniform family of oracle circuits $\cC_n$ 
can be implemented by a uniform family of  arithmetic circuits 
of polynomial size. 
This completes the proof of~\cref{th:MAIN-u}. 

\begin{rem} 
The vanishing order $e_j$ of $R_j$ 
at $t=0$, that is $R_j  = \tc_t(R_j)t^{e_j} + O(t^{e_j+1})$,
may be exponential in $n$, see~\cref{re:order-e}. 
It is therefore unclear how to extract $\tc_t(R_j)$ 
in polynomial time from $R_j$.
On the other hand, one can show that the vanishing order of $\tc_w(R_j)$ 
at $w=0$ is at most~$n$. 
For this reason, $\rho_j$ could be computed by interpolation 
from $\tc_w(R_j)$ in polynomial time.
\end{rem}


\section{Outlook and further research}\label{se:outlook}

There are various interesting directions to pursue further. 
In~\cite{andrews-et-al-HN:26} it is also shown that the problem of counting the number of complex solutions 
of a given system of polynomial equations (with integer coefficients, say) can be done in polynomial time
with oracle calls to the counting hierarchy. This can be expressed as the inclusion 
$\GCC\subseteq {\sf \Po}^{\CH}$ in terms of the complexity class~$\GCC$, 
introduced in~\cite{bucu-count-II:06}. 
The same paper also defined the counting complexity class $\SP^0_\C$ in the Blum-Shub-Smale-model over $\C$.
One may ask whether the result on counting complex solutions in~\cite{andrews-et-al-HN:26} has 
an analogue in the Blum-Shub-Smale-model over~$\C$.
It looks plausible that problems in $\SP^0_\C$ can be solved by uniform families of arithmetic circuits of polynomially size, 
endowed with oracle gates for evaluating multivariate resultants and gates for computing trailing coefficients.
If this is the case, then $\VP^0(\U) =\VNP^0(\U)$ would imply that all problems in 
$\SP^0_\C$ can be solved in polynomial time by constant-free BSS-machines over $\C$. 

It may also be possible to relate the polynomial hierarchy $\PH^0_\C$ in the BSS-model over $\C$ 
to the algebraic complexity of the permanent, arriving at a result in the spirit of 
Toda's theorem; see \cite[Problem~8.2]{bucu-quaderni:04}

The most exciting open question is to what extent the above ideas can be extended to computations with real numbers.
In this context, the problem $\HN_\C$ is to be replaced by the problem 
of deciding whether a given multivariate 
polynomial has a real zero. This problem is $\NP^0_\R$-complete~\cite{blss:89,BCSS:98}.
When restricting to polynomials with integer coefficients, the complexity of 
this feasibility problem can be studied in the model of Turing machines. 
This feasibility problem is then complete 
for the complexity class $\exists\R$, 
called \emph{existential theory of the reals}, 
which 
is the same as the Boolean part of the class $\NP^0_\R$   
in the terminology in~\cite{bucu-count-II:06}. 
Many interesting problems in computational geometry are complete for $\exists\R$, 
among those is 
the feasibility problem for systems of quadratic forms; see~\cite[A.6]{ER-survey}.

In the arXiv version of this paper (23 June 2026) I raised the following question.

\begin{quest}\label{conj:ER-in-CH}
Is the existential theory of the reals contained in the counting hierarchy?
\end{quest}

In October 2026, several ChatGPT generated papers have appeared, 
claiming an affirmative solution to this question.
None of these papers has been humanly verified so far. See
\cite{meiburg:26} and \cite{oai-etr:26}.
%


\section*{Acknowledgments} 

I am grateful to Robert Andrews for sharing the ideas in his breakthrough work~\cite{andrews-et-al-HN:26}, 
while it was still in progress. Also, I thank Robert Andrews, Pascal Koiran, and Guillaume Malod 
for feedback on my manuscript. 
This work was triggered by discussions at the inspiring semester program ``Complexity and Linear Algebra''
in the Fall of 2025 at the Simons Institute for the Theory of Computing in Berkeley. 
I thank the Simons Institute for making this program possible. 

\newcommand{\etalchar}[1]{$^{#1}$}
\def\cprime{$'$}

\end{document}